\documentclass[submission,copyright,creativecommons]{eptcs}
\providecommand{\event}{WPTE 2017} 
\usepackage{amssymb}
\usepackage{color}

\DeclareMathAlphabet{\mathcal}{OMS}{cmsy}{m}{n}

\usepackage{theorem}

\newcommand{\eop}{\hspace*{\fill}$\Box$}
\def\qed{\eop}
\newtheorem{definition}{Definition}[section]
\newtheorem{lemma}[definition]{Lemma}
\newtheorem{theorem}[definition]{Theorem}
\newtheorem{corollary}[definition]{Corollary}
{\theorembodyfont{\rmfamily}
  \newtheorem{example}[definition]{\it Example}
  \newtheorem{proof}{\it Proof.}
}

\def\cC{\mathcal{C}}
\def\cD{\mathcal{D}}
\def\cF{\mathcal{F}}
\def\cG{\mathcal{G}}
\def\cM{\mathcal{M}}
\def\cP{\mathcal{P}}
\def\cR{\mathcal{R}}
\def\cS{\mathcal{S}}
\def\cV{\mathcal{V}}
\def\Root{\mathit{root}}
\def\Var{{\mathcal{V}\!\mathit{ar}}}
\def\Dom{{\mathcal{D}\mathit{om}}}
\def\Range{{\mathcal{R}\mathit{an}}}
\def\symb#1{\mathsf{#1}}
\def\Subst{{\mathcal{S}\mathit{ub}}}
\def\Hole{\Box}
\def\Rule#1#2{#1 \to #2}

\def\Condition#1#2{#1 \tto #2}
\def\Int{\mathbb{Z}}
\def\cPred{\cP}
\def\EqualityPred{\simeq}
\def\limplies{\Rightarrow}
\def\Constrained#1{[\hspace{-2pt}[\,#1\,]\hspace{-2pt}]}
\def\CRule#1#2#3{#1 \to #2 ~\Constrained{#3}}
\def\FVar{\Var}
\def\Fol{\mathcal{F}\!\mathit{ol}}
\def\LIA{\mathsf{LIA}}
\def\cGlia{\cG_\LIA}
\def\cPlia{\cPred_\LIA}
\def\cMlia{\cM_\LIA}
\def\cGint{\cG_\Int}
\def\cPint{\cP_\Int}
\def\cMint{\cM_\Int}
\def\Pol{\mathcal{P}\!\mathit{ol}}
\def\DP{\mathit{DP}}
\def\DPproblem#1#2{#1}
\def\Marked#1{#1^\sharp}
\def\Proc{\mathit{Proc}}
\def\PIProc{\Proc_{\mathsf{PI}}}
\def\Pgtr{\cS_>}
\def\Pbound{\cS_{\mathsf{bound}}}
\def\Pfilter{\cS_{\mathsf{filter}}}
\def\rel{\bowtie}
\newcommand{\newPgtr}[1][\rel]{\cS_{#1}}
\def\DecDec{(>,\geq,\geq)}
\def\DecInc{(>,\leq,\leq)}
\def\IncDec{(<,\geq,\leq)}
\def\IncInc{(<,\leq,\geq)}
\newcommand{\PIProcX}[1][(\rel_1,\rel_2,\rel_3)]{\Proc_{#1}}
\def\Condition#1{(\textbf{#1})}
\def\CondInterpreted{\Condition{A1}}
\def\CondPreserve{\Condition{A2}}
\def\CondSubtraction{\Condition{A3}}
\def\CondPolF{\Condition{A4}}
\def\CondRdec{\Condition{R1}}
\def\CondRinc{\Condition{R2}}
\def\CondPolMdec{\Condition{P1}}
\def\CondPolMinc{\Condition{P2}}
\def\CondSdec{\Condition{S1}}
\def\CondSinc{\Condition{S2}}
\def\AProVE{\textsf{AProVE}}
\def\cRconst{\cR_1}
\def\cRsp{\cR_0}
\def\cRMc{\cR_1}
\def\cRack{\cR_2}
\def\cRnest{\cR_3}
\def\Success{success}

\title{Transforming Dependency Chains of Constrained TRSs into Bounded Monotone Sequences of Integers}
\def\titlerunning{Transforming Dependency Chains 
into Bounded Monotone Sequences of Integers}
\author{Tomohiro Sasano
\institute{Graduate School of Information Science\\
Nagoya University\\
Nagoya, Japan}
\email{sasano@trs.cm.is.nagoya-u.ac.jp}
\and
Naoki Nishida 
\institute{Graduate School of Informatics\\
Nagoya University\\
Nagoya, Japan}
\email{nishida@i.nagoya-u.ac.jp}
\and
Masahiko Sakai
\institute{Graduate School of Informatics\\
Nagoya University\\
Nagoya, Japan}
\email{sakai@i.nagoya-u.ac.jp}
\and
Tomoya Ueyama
\institute{Graduate School of Information Science\\
Nagoya University\\
Nagoya, Japan}
}
\def\authorrunning{T.~Sasano, N.~Nishida, M.~Sakai, and T.~Ueyama}
\def\copyrightholders{Sasano, Nishida, Sakai, and Ueyama}

\begin{document}
\maketitle

\begin{abstract}
In the dependency pair framework for proving termination of rewriting systems, polynomial interpretations are used to transform dependency chains 
into bounded decreasing sequences of integers, and they play an important role for the success of proving termination, especially for constrained rewriting systems.
In this paper, we show sufficient conditions of linear polynomial interpretations for  
transforming dependency chains into bounded monotone (i.e., decreasing or increasing) sequences of integers.
Such polynomial interpretations transform rewrite sequences of the original system into decreasing or increasing sequences independently of the transformation of dependency chains.
When we transform rewrite sequences into increasing sequences, polynomial interpretations have non-positive coefficients for reducible positions of marked function symbols.
We propose four DP processors parameterized by transforming dependency chains and rewrite sequences into either decreasing or increasing sequences of integers, respectively.
We show that such polynomial interpretations make us succeed in proving termination of the McCarthy 91 function over the integers. 
\end{abstract}




\section{Introduction}
\label{sec:intro}

Recently, techniques developed for term rewriting systems (TRSs, for short)
have been applied to the verification of
programs written in several programming languages %
(cf.~\cite{FKN17tocl}).
In verifying programs with comparison operators over the integers via term rewriting,
\emph{constrained rewriting} is very useful to avoid very 
complicated rewrite rules for the comparison operators, and various formalizations of
constrained rewriting have been proposed:
\emph{constrained TRSs}~\cite{FNSKS08b,BJ08,SNSSK09,SNS11} 
(e.g., \emph{membership conditional TRSs}~\cite{Toy87}),
\emph{constrained equational systems} (CESs, for short)
~\cite{FK08},
\emph{integer TRSs} (ITRSs, for short)~\cite{FGPSF09},
\emph{PA-based TRSs} ($\Int$-TRSs)~\cite{FK09} (simplified variants of CESs), and
\emph{logically constrained TRSs} (LCTRSs, for short)~\cite{KN13frocos,KN14aplas}.

One of the most important properties that are often verified in practice is \emph{termination}, and many methods for proving termination have been developed in the literature,
especially in the field of term rewriting
(cf.~the survey of Zantema~\cite{Zan03}).
At present, the \emph{dependency pair} (DP) \emph{method}~\cite{AG00} and the
\emph{DP framework}~\cite{GTS04} are key fundamentals for proving termination of TRSs, and they have been extended to several kinds of rewrite
systems~\cite{FK08,AEFGGLST08,FK09,FGPSF09,SNS11,LM14,FKN17tocl}.
In the DP framework, termination problems are reduced to finiteness of \emph{DP problems}
which consist of sets of dependency pairs and rewrite rules. 
We prove finiteness by applying \emph{sound DP processors} to an input DP problem and then by
decomposing the DP problem to smaller ones in the sense that all the DP sets of output DP problems
are strict subsets of the DP set of the input problem.
In the DP frameworks for constrained rewriting~\cite{FK08,SNS11}, the DP processors based on \emph{polynomial interpretations} (the PI-based processors, for short) decompose a given DP problem by using a polynomial interpretation $\Pol$ that transforms dependency chains into bounded decreasing sequences of integers---roughly speaking, a dependency pair $\CRule{\Marked{s}}{\Marked{t}}{\varphi}$ is removed from the given problem if the integer arithmetic formula $\varphi \Rightarrow \Pol(s) > \Pol(t)$ is valid. 
The processor in~\cite{SNS11} can be considered a simplified version of that in~\cite{FK08} in the sense that for efficiency, $\Pol$ drops \emph{reducible positions}---arguments of marked symbols, which may contain an uninterpreted function symbol when a dependency pair is instantiated---and then the rules in the given system can be ignored in applying the PI-based processor.
Such a simplification is not so restrictive when we prove termination of \emph{counter-controlled loops}, e.g., 
\verb|for(i=0;i<n;i++){|\,\ldots \verb|}|.
However, the simplification sometimes prevents us from proving termination of a function, the definition of which has nested function calls.


Let us consider the following constrained TRS defining the \emph{McCarthy 91 function}:
\[
  \cRMc =
 \left\{
 \begin{array}{r@{~~~~}r@{\>}c@{\>}l@{~~}c@{}c@{}c}
  (1) & \CRule{\symb{f}(x) &}{& \symb{f}(\symb{f}(\symb{s}^{11}(x))) &}{& \symb{s}^{101}(\symb{0}) > x &} \\
  (2) & \CRule{\symb{f}(x) &}{& \symb{p}^{10}(x) &}{& \neg (\symb{s}^{101}(\symb{0}) > x) &} \\
 \end{array}
 \right\}
 \cup
  \cRsp
\]
where
$\cRsp
=
\{~
  \Rule{\symb{s}(\symb{p}(x))}{x}, ~~
  \Rule{\symb{p}(\symb{s}(x))}{x}
~\}$.
It is known that the function always terminates and returns $91$ if an integer $n \leq 101$ is given as input, and $n-10$ otherwise: 
$\forall n \in \Int.\ ({n \leq 101} \Rightarrow {\symb{f}(n)=91}) \land ({n > 101} \Rightarrow {\symb{f}(n)=n-10})$.
Termination of the McCarthy 91 function can be proved automatically if the function is defined over the natural numbers~\cite{Giesl97}.
However, the method in~\cite{Giesl97} cannot prove termination of the function that is defined over the integers.
As another approach, let us consider the DP framework. 
The dependency pairs of $\cRMc$ are:
\[
\begin{array}{@{}l@{\>}l@{}}
  \DP(\cRMc) = &
 \left\{
 \begin{array}{@{~~~\,}r@{~~~~}r@{\>}c@{\>}l@{~~}c@{}c@{}c}
  (3) & \CRule{\Marked{\symb{f}}(x) &}{& \Marked{\symb{f}}(\symb{f}(\symb{s}^{11}(x))) &}{& \symb{s}^{101}(\symb{0}) > x &} \\
  (4) & \CRule{\Marked{\symb{f}}(x) &}{& \Marked{\symb{f}}(\symb{s}^{11}(x)) &}{& \symb{s}^{101}(\symb{0}) > x &} \\
 \end{array}
 \right\} \\[5pt]
 & {} \cup 
 \{~ (5) ~~~~ \CRule{\Marked{\symb{f}}(x)}{\Marked{\symb{s}}(\symb{s}^i(x)) ~~~~~~~\,}{\symb{s}^{101}(\symb{0}) > x} \mid 0 \leq i \leq 10 ~\} 
 \\
 & {} \cup 
 \{~ (6) ~~~~ \CRule{\Marked{\symb{f}}(x)}{\Marked{\symb{p}}(\symb{p}^i(x)) ~~~~~~\,}{\neg (\symb{s}^{101}(\symb{0}) > x} \mid 0 \leq i \leq 9 ~\} 
 \\
 \end{array}
\]
Let us focus on (3) and (4) because no infinite chain of dependency pairs contains any of (5) and (6).
Unfortunately, the method in~\cite{SNS11} for proving termination of constrained TRSs cannot prove termination of $\cRMc$ because the right-hand side of (3) is of the form $\Marked{\symb{f}}(\symb{f}(\symb{s}^{11}(x)))$ and thus we have to drop the first argument of $\Marked{f}$, i.e., $\Pol(\Marked{f}) = a_0$ where $a_0$ is an integer.
Both sides of (3) and (4) are converted by $\Pol$ to $a_0$ and we do not remove any of (3) and (4).
To make the method in~\cite{SNS11} more powerful, let us allow $\Pol(\Marked{f})$ to keep its reducible positions as in~\cite{FK08}.
Then, for $\cRMc$, $\Pol$ has to be an interpretation over the natural numbers, and for each rule $\CRule{\ell}{r}{\varphi}$ in $\cRMc$ the validity of the integer arithmetic formula $\varphi \Rightarrow \Pol(\ell) \geq \Pol(r)$ is required. 
However, such an interpretation does not exist for $\cRMc$.

In this paper, 
we extend the PI-based processor in~\cite{SNS11} by making its linear polynomial interpretation $\Pol$ transform dependency chains into bounded \emph{monotone} (i.e., decreasing or increasing) sequences of integers.
To be more precise, given a constrained TRS $\cR$, 
\begin{itemize}
	\item for function symbols in $\cR$, $\Pol$ is an interpretation over the natural numbers as in~\cite{FK08}, while constants that are not coefficients may be negative integers (i.e., for $\Pol(f)=b_0 + b_1x_1 + \cdots + b_n x_n$, the coefficients $b_1,\ldots,b_n$ have to be non-negative integers but the constant $b_0$ may be a negative integer),
	\item for rules in $\cR$, we require one of the following:
		\begin{itemize}
			\item[\CondRdec] $\varphi \Rightarrow \Pol(\ell) \geq \Pol(r)$ is valid for all $\CRule{\ell}{r}{\varphi} \in \cR$ (i.e., rewrite sequences of $\cR$ are transformed by $\Pol$ into \emph{decreasing} sequences of integers),
				or
			\item[\CondRinc] $\varphi \Rightarrow \Pol(\ell) \leq \Pol(r)$ is valid for all $\CRule{\ell}{r}{\varphi} \in \cR$ (i.e., rewrite sequences of $\cR$ are transformed by $\Pol$ into \emph{increasing} sequences of integers),
			and
		\end{itemize}
	\item for monotonicity of transformed sequences, coefficients for \emph{reducible positions} of \emph{marked} symbols have to satisfy a sufficient condition---to be non-negative for {\CondRdec} and to be non-positive for {\CondRinc}---and the second argument of the subtraction symbol (i.e., ``$-$'') is an interpretable term in anywhere.
\end{itemize}
Such a polynomial interpretation transforms all dependency chains into bounded decreasing sequences of integers, or all to bounded increasing sequences of integers.
Since we have two possibilities for transforming rewrite sequences of $\cR$, we have four kinds of PI-based processors: $\PIProcX[\DecDec]$, $\PIProcX[\DecInc]$, $\PIProcX[\IncDec]$, and $\PIProcX[\IncInc]$ in Table~\ref{tbl:overview}.
Then, we show an experimental result to compare the four PI-based processors by using them to prove termination of 
some constrained TRSs. 
Although this paper adopts the class of constrained TRSs in~\cite{FNSKS08b,SNSSK09,SNS11}, it would be straightforward to adapt our results to other higher-level styles of constrained systems in, e.g.,~\cite{FK08,FK12,KN13frocos}.
It would also be straightforward to extend the results for the single-sorted
case to the many-sorted one (cf.~\cite{Kop13}).

\begin{table}[tb]
\caption{our transformations of ground dependency chains into monotone sequences of integers.}
\label{tbl:overview}
\centering
\small
\begin{tabular}{@{~}c|c@{\hspace{-2ex}}c@{\hspace{-2ex}}c@{\hspace{-1.5ex}}c@{\hspace{-1.5ex}}c@{\hspace{-2ex}}c@{\hspace{-2ex}}c@{\hspace{-1.5ex}}c@{\hspace{-1.5ex}}c@{\hspace{-2ex}}c@{~}c@{~}}
\hline
\rule{0pt}{13pt}
chain of $\cR$
& $\Marked{f_0}(s_0)$ & $\to_{\varepsilon,\DP(\cR)}$ & $\Marked{f_1}(t_0)$ & $\to^*_{>\varepsilon,\cR}$
& $\Marked{f_1}(s_1)$ & $\to_{\varepsilon,\DP(\cR)}$ & $\Marked{f_2}(t_1)$ & $\to^*_{>\varepsilon,\cR}$
& $\Marked{f_2}(s_2)$ & $\to_{\varepsilon,\DP(\cR)}$ 
& $\cdots$ \\[-2pt]
& & & $\vdots$ & 
& $\vdots$ & & $\vdots$ & 
& $\vdots$ & 
& \\[-2pt]
$\cR$-steps
& & & $t_0$ & $\to^*_{\cR}$
& $s_1$ & & $t_1$ & $\to^*_{\cR}$
& $s_2$ & 
& $\cdots$ \\[2pt]
\hline\hline
\rule{0pt}{13pt}
$\PIProc$~\cite{SNS11}
& $\Pol(\Marked{f_0}(s_0))$ & $\geq$ & $\Pol(\Marked{f_1}(t_0))$ & $=$
& $\Pol(\Marked{f_1}(s_1))$ & $\geq$ & $\Pol(\Marked{f_2}(t_1))$ & $=$
& $\Pol(\Marked{f_2}(s_2))$ & $\geq$ 
& $\cdots$ \\[-2pt]
(Def.~\ref{def:PIProc}) & & & $\vdots$ & 
& $\vdots$ & & $\vdots$ & 
& $\vdots$ & 
& \\
& & & $\Pol(t_0)$ & $=$
& $\Pol(s_1)$ &  & $\Pol(t_1)$ & $=$
& $\Pol(s_2)$ &  
& $\cdots$ \\[2pt]
\hline\hline
\rule{0pt}{13pt}
$\PIProcX[\DecDec]$ 
& $\Pol(\Marked{f_0}(s_0))$ & $\geq$ & $\Pol(\Marked{f_1}(t_0))$ & $\geq$
& $\Pol(\Marked{f_1}(s_1))$ & $\geq$ & $\Pol(\Marked{f_2}(t_1))$ & $\geq$
& $\Pol(\Marked{f_2}(s_2))$ & $\geq$ 
& $\cdots$ \\[-2pt]
($\PIProc+$\cite{FK08})
& & & $\vdots$ & 
& $\vdots$ & & $\vdots$ & 
& $\vdots$ & 
& \\
(Sec.~\ref{subsec:dec-R}) & & & $\Pol(t_0)$ & $\geq$
& $\Pol(s_1)$ &  & $\Pol(t_1)$ & $\geq$
& $\Pol(s_2)$ &  
& $\cdots$ \\[2pt]
\hline
\rule{0pt}{13pt}
$\PIProcX[\DecInc]$
& $\Pol(\Marked{f_0}(s_0))$ & $\geq$ & $\Pol(\Marked{f_1}(t_0))$ & $\geq$
& $\Pol(\Marked{f_1}(s_1))$ & $\geq$ & $\Pol(\Marked{f_2}(t_1))$ & $\geq$
& $\Pol(\Marked{f_2}(s_2))$ & $\geq$ 
& $\cdots$ \\[-2pt]
(Sec.~\ref{subsec:inc-R})
& & & $\vdots$ & 
& $\vdots$ & & $\vdots$ & 
& $\vdots$ & 
& \\
& & & $\Pol(t_0)$ & \textcolor{blue}{$\leq$}
& $\Pol(s_1)$ &  & $\Pol(t_1)$ & \textcolor{blue}{$\leq$}
& $\Pol(s_2)$ &  
& $\cdots$ \\[2pt]
\hline
\rule{0pt}{13pt}
$\PIProcX[\IncDec]$
& $\Pol(\Marked{f_0}(s_0))$ & \textcolor{blue}{$\leq$} & $\Pol(\Marked{f_1}(t_0))$ & \textcolor{blue}{$\leq$}
& $\Pol(\Marked{f_1}(s_1))$ & \textcolor{blue}{$\leq$} & $\Pol(\Marked{f_2}(t_1))$ & \textcolor{blue}{$\leq$}
& $\Pol(\Marked{f_2}(s_2))$ & \textcolor{blue}{$\leq$} 
& $\cdots$ \\[-2pt]
(Sec.~\ref{subsec:inc-DP})
& & & $\vdots$ & 
& $\vdots$ & & $\vdots$ & 
& $\vdots$ & 
& \\
& & & $\Pol(t_0)$ & $\geq$
& $\Pol(s_1)$ &  & $\Pol(t_1)$ & $\geq$
& $\Pol(s_2)$ &  
& $\cdots$ \\[2pt]
\hline
\rule{0pt}{13pt}
$\PIProcX[\IncInc]$
& $\Pol(\Marked{f_0}(s_0))$ & \textcolor{blue}{$\leq$} & $\Pol(\Marked{f_1}(t_0))$ & \textcolor{blue}{$\leq$}
& $\Pol(\Marked{f_1}(s_1))$ & \textcolor{blue}{$\leq$} & $\Pol(\Marked{f_2}(t_1))$ & \textcolor{blue}{$\leq$}
& $\Pol(\Marked{f_2}(s_2))$ & \textcolor{blue}{$\leq$} 
& $\cdots$ \\[-2pt]
(Sec.~\ref{subsec:inc-DP})
& & & $\vdots$ & 
& $\vdots$ & & $\vdots$ & 
& $\vdots$ & 
& \\
& & & $\Pol(t_0)$ & \textcolor{blue}{$\leq$}
& $\Pol(s_1)$ &  & $\Pol(t_1)$ & \textcolor{blue}{$\leq$}
& $\Pol(s_2)$ &  
& $\cdots$ \\[2pt]
\hline
\end{tabular}
\end{table}

The contribution of this paper is 
to develop a technique to automatically prove termination of the McCarthy 91 function via linear polynomial interpretations that transform dependency chains into bounded monotone (i.e., not only decreasing but also increasing) sequences of integers, and that transform rewrite sequences of the given constrained TRS into monotone sequences of integers.

This paper is organized as follows.
In Section~\ref{sec:preliminaries}, we briefly recall the basic
notions and notations of constrained rewriting.
In Section~\ref{sec:DP-framework}, we briefly recall the DP method
for constrained TRSs. 
In Section~\ref{sec:improvement}, 
we show an improvement of the PI-based processor and also show results of experiments to evaluate the proposed PI-based processors.
In Section~\ref{sec:conclusion}, we conclude this paper and describe related work and future work of this research.

\section{Preliminaries}
\label{sec:preliminaries}


In this section, we briefly recall the basic notions and notations of term
rewriting~\cite{BN98,Ohl02}, and constrained
rewriting~\cite{FNSKS08b,BJ08,SNSSK09,SNS11}.

Throughout the paper, we use $\cV$ as a countably infinite set of
\emph{variables}.
We denote the set of \emph{terms} over a
signature $\Sigma$ and a variable set $V \subseteq \cV$ by $T(\Sigma,V)$. 
We often write $f/n$ to represent an $n$-ary symbol $f$.
We abbreviate the set $T(\Sigma,\emptyset)$ of \emph{ground terms} over $\Sigma$
to $T(\Sigma)$. 
We denote the set of variables appearing in a term $t$ by $\Var(t)$.
A \emph{hole} $\Hole$ is a special constant not appearing in considered
signatures (i.e., $\Hole\notin\Sigma$), and a term in $T(\Sigma\cup\{\Hole\},V)$ is
called a
\emph{context} over $\Sigma$ and $V$ if the hole $\Hole$ appears in the term
exactly once. 
We denote the set of contexts over $\Sigma$ and $V$ by $T_\Hole(\Sigma,V)$. 
For a term $t$ and a context $C[~]_p$ with the hole at a \emph{position}
$p$, we denote by $C[t]_p$ the term obtained from $t$ and $C[~]_p$ by
replacing the hole at $p$ by $t$.
We may omit $p$ from $C[~]_p$ and $C[t]_p$. 
For a term $C[t]_p$, the term $t$ is a \emph{subterm} of $C[t]$
(at $p$). 
Especially, when $p$ is not the \emph{root} position $\varepsilon$, we
call $t$ a \emph{proper subterm} of $C[t]$.
For a term $s$ and a position $p$ of $s$, we denote the subterm of $s$
at $p$ by $s|_p$, and the function symbol at the root position of $s$ by $\Root(s)$. 

The \emph{domain} and \emph{range} of a \emph{substitution} $\sigma$ are
denoted by $\Dom(\sigma)$ and $\Range(\sigma)$, respectively.
For a signature $\Sigma$, a substitution $\sigma$ is called \emph{ground}
if $\Range(\sigma) \subseteq T(\Sigma)$.
For a subset $V$ of $\cV$, we
denote the set of substitutions over $\Sigma$ and $V$ by $\Subst(\Sigma,V)$:
$\Subst(\Sigma,V) = \{ \sigma \mid \Range(\sigma) \subseteq T(\Sigma,V)\}$.
We abbreviate $\Subst(\Sigma,\emptyset)$ to $\Subst(\Sigma)$.
We may write $\{ x_1 \mapsto t_1, ~\ldots, ~x_n \mapsto t_n \}$ instead of
$\sigma$ if $\Dom(\sigma)=\{x_1,\ldots,x_n\}$ and $\sigma(x_i)=t_i$ for
all $1 \leq i \leq n$.
We may write $t\sigma$ for the application $\sigma(t)$ of $\sigma$ to
$t$. 
For a subset $V$ of $\cV$, we denote the restricted substitution of
$\sigma$ w.r.t.\ $V$ by $\sigma|_V$:
$\sigma|_V = \{ x \mapsto \sigma(x) \mid x \in \Dom(\sigma)\cap V\}$.


Let $\cG$ be a signature (e.g., a subsignature of $\Sigma$) and $\cPred$ a set of \emph{predicate
symbols}, each of which has a fixed arity, 
and $\cM$ a \emph{structure}
specifying interpretations for symbols in $\cG$ and $\cPred$:
$\cM$ has a \emph{universe} (a non-empty set), and $g^\cM$ and $p^\cM$
are interpretations for a function symbol $g\in\cG$ and a predicate
symbol $p\in\cPred$, respectively.
Ground terms in $T(\cG)$ are interpreted by $\cM$ in the usual way.
We use $\top$ and $\bot$ for Boolean values \emph{true} and
\emph{false},%
\footnote{
Note that $\top$ and $\bot$ are just symbols used in e.g., constraints of rewrite rules, and we distinguish them with \emph{true} and \emph{false} used as values.}
 and usual logical connectives $\neg$, $\vee$, $\land$, and $\limplies$, 
which are interpreted in the usual way. 
For the sake of simplicity, we do not use quantifiers in formulas. 
We assume that $\cPred$ contains a binary symbol $\EqualityPred$ for
\emph{equality}. 
For a subset $V\subseteq \cV$, we
denote the set of \emph{formulas} over $\cG$, $\cPred$, and $V$ by
$\Fol(\cG,\cPred,V)$. 
The set of variables in a formula $\varphi$ is denoted by $\FVar(\varphi)$.
Formulas in $\Fol(\cG,\cPred,\cV)$ are called \emph{constraints} (w.r.t.\ $\cM$).
We assume that for each element $a$ in the universe, there exists a
ground term $t$ in $T(\cG)$ such that $t^\cM=a$.
A ground formula $\varphi$ is said to \emph{hold w.r.t\ $\cM$}, 
written as $\cM \models \varphi$, if $\varphi$ is interpreted by $\cM$ as \emph{true}.
The application of a substitution $\sigma \in\Subst(\cG,\cV)$ is
naturally extended to formulas, and $\sigma(\varphi)$ is abbreviated to
$\varphi\sigma$.
Note that for a signature $\Sigma$ with $\cG\subseteq\Sigma$, we cannot apply $\sigma$ to $\varphi \in \Fol(\cG,\cPred,\cV)$ if $\sigma|_{\FVar(\varphi)} \notin \Subst(\cG,\cV)$.%
\footnote{
When considering formulas in $\Fol(\cG,\cPred,\cV)$, we force $\sigma\varphi$ to be in $\Fol(\cG,\cPred,\cV)$.}
A formula $\varphi$ is called \emph{valid w.r.t.\ $\cM$}
(\emph{$\cM$-valid}, for short) if $\cM \models \varphi\sigma$ for all ground
substitutions $\sigma\in\Subst(\cG)$ with
$\FVar(\varphi)\subseteq\Dom(\sigma)$, 
and called \emph{satisfiable w.r.t.\ $\cM$}
(\emph{$\cM$-satisfiable}, for short) if $\cM \models \varphi\sigma$ for some ground
substitution $\sigma\in\Subst(\cG)$ such that
$\FVar(\varphi)\subseteq\Dom(\sigma)$. 
A structure $\cM$ for $\cG$ and $\cPred$ is called an \emph{LIA-structure}
if 
 the universe is the integers,
 every symbol $g \in \cG$ is interpreted as a linear integer
       arithmetic expression, 
       and
 every symbol $p \in \cPred$ is interpreted as a 
       Presburger arithmetic sentence over the integers, e.g., binary
       comparison predicates.

Let $\cF$ and $\cG$ be pairwise disjoint signatures (i.e., $\cF\cap\cG =
\emptyset$),%
\footnote{
A signature $\Sigma$ is explicitly divided
into $\cF$ and $\cG$ (i.e., $\Sigma=\cF\uplus\cG$) where $\cF$ is the set of
\emph{uninterpreted} symbols and $\cG$ the set of \emph{interpreted} symbols.
To make this distinguish clear, we always separate $\cF$ and $\cG$, e.g., we write $(\cF,\cG)$ but not $\cF\uplus\cG$.
}
$\cPred$ a set of predicate symbols, and $\cM$ a structure for
$\cG$ and $\cPred$. 
A \emph{constrained rewrite rule} over $(\cF,\cG,\cPred,\cM)$
is a triple $(\ell,r,\varphi)$, denoted by $\CRule{\ell}{r}{\varphi}$, such that
$\ell,r \in T(\cF\cup\cG,\cV)$, $\ell$ is not a variable, $\varphi \in
\Fol(\cG,\cPred,\cV)$, and $\Var(\ell) \supseteq \Var(r)\cup\FVar(\varphi)$. 
We usually consider $\cM$-satisfiable constraints for $\varphi$.
When $\varphi$ is $\top$, we may write $\Rule{\ell}{r}$ instead of
$\CRule{\ell}{r}{\top}$. 
A \emph{constrained term rewriting system} (constrained TRS, for short) over
$(\cF,\cG,\cPred,\cM)$ is a finite set $\cR$ of constrained
rewrite rules over $(\cF,\cG,\cPred)$. 
When 
$\varphi=\top$ for all rules
$\CRule{\ell}{r}{\varphi}\in\cR$, 
$\cR$ is a \emph{term rewriting system} (TRS, for short).
The \emph{rewrite relation} $\to_\cR$ of $\cR$ is defined as follows:
$\to_\cR = \{ (C[\ell\sigma]_p, C[r\sigma]_p) 
 \mid \CRule{\ell}{r}{\varphi} \in \cR,
 ~ C[~] \in T_\Hole(\cF\cup\cG,\cV),
 ~ \sigma \in \Subst(\cF\cup\cG,\cV),
 ~ \sigma|_{\FVar(\varphi)} \in \Subst(\cG,\cV),
 ~\mbox{$\varphi\sigma$ is $\cM$-valid} \}$.
To specify the position $p$ where the term is reduced,
we may write $\to_{p,\cR}$ or $\to_{>q,\cR}$ where $p > q$.
A term $t$ is called \emph{terminating} (w.r.t.\ $\cR$) if there is no infinite reduction sequence $t \to_\cR t_1 \to_\cR t_2 \to_\cR \cdots$. 
$\cR$ is called \emph{terminating} if every term is terminating. 
For a constrained TRS $\cR$ over $(\cF,\cG,\cPred,\cM)$, 
the sets $\cD_\cR$ and $\cC_\cR$ of \emph{defined symbols} and
\emph{constructors}, respectively, are defined as follows: 
$\cD_\cR = \{ f \in \cF\cup\cG \mid \CRule{f(t_1,\ldots,t_n)}{r}{\varphi} \in
\cR\}$ and $\cC_\cR = (\cF\cup\cG) \setminus \cD_\cR$.

\begin{example}
\label{ex:Glia}
Let $\cGlia = \{ \symb{0}/0, \symb{s}/1, \symb{p}/1
\}$, 
$\cPlia = \{ {\EqualityPred}, {>}, {\geq} 
\}$, 
and 
$\cMlia$ an LIA-structure for $\cGlia$ and $\cPlia$ such that
the universe is $\Int$,
$\symb{0}^{\cMlia} = 0$,
       $\symb{s}^{\cMlia}(x) = x + 1$,
       $\symb{p}^{\cMlia}(x) = x - 1$,
       and $>$ and $\geq$
are interpreted as the corresponding comparison predicates
 in the usual way.
Then, we have that 
 $(\symb{s}(\symb{s}(\symb{0})))^{\cMlia} = 2$,
 $(\symb{s}(\symb{p}(\symb{p}(\symb{s}(\symb{0})))))^{\cMlia} = 0$, and
 so on.
$\cRconst$ in Section~\ref{sec:intro} is
 over $(\{\symb{f}/1\},\cGlia,\cPlia,\cMlia)$, 
 and we have e.g., 
$
   \symb{f}(\symb{s}^{100}(\symb{0}))
    \mathrel{\to_{\cRMc}}
    \symb{f}(\symb{f}(\symb{s}^{111}(\symb{0})))
    \mathrel{\to_{\cRMc}}
    \symb{f}(\symb{p}^{10}(\symb{s}^{111}(\symb{0})))
	\mathrel{\to^*_{\cRMc}}
    \symb{f}(\symb{s}^{101}(\symb{0}))
	\mathrel{\to_{\cRMc}}
    \symb{p}^{10}(\symb{s}^{101}(\symb{0}))
	\mathrel{\to^*_{\cRMc}}
    \symb{s}^{91}(\symb{0})
$. 
\end{example}

We assume that $\cR$ is \emph{locally sound for
$\cM$}~\cite{SNSSK09,SNS11}, i.e., for every 
rule $\CRule{\ell}{r}{\varphi} \in \cR$,
if the root symbol of $\ell$ is in $\cG$, 
 then
$r$ and all the proper subterms of $\ell$ are in $T(\cG,\cV)$, 
       and
the formula $\varphi \limplies (\ell \mathop{\EqualityPred} r)$ is $\cM$-valid.
Local soundness for $\cM$ ensures consistency for the semantics and further
that no interpreted ground term is reduced to any term containing an uninterpreted 
function symbol.
This property is implicitly assumed in other formalizations of constrained 
rewriting by e.g., rules for constructors are separated from user-defined 
rules~\cite{BJ08,FK08},
or rules are defined for uninterpreted function symbols only~\cite{KN13frocos}.

\section{The DP Framework for Constrained TRSs}
\label{sec:DP-framework}

In this section, we recall the \emph{DP framework}
for constrained TRSs~\cite{SNS11}, which is a straightforward extension of the DP
framework~\cite{GTS04,FK08} for TRSs to constrained TRSs.  

In the following, 
we let $\cR$ be a constrained TRS over $(\cF,\cG,\cPred,\cM)$
unless noted otherwise. 
We introduce a \emph{marked symbol}
$\Marked{f}$ for each defined symbol $f$ of $\cR$, where $\Marked{f}\notin \cF\cup\cG$.
We denote the set of marked symbols by $\Marked{\cD_\cR}$.
For a term $t$ of the form $f(t_1,\ldots,t_n)$ in $T(\cF\cup\cG,\cV)$ with $f/n \in \cD_\cR$,
we denote $\Marked{f}(t_1,\ldots,t_n)$ (a marked term) by
$\Marked{t}$. 
To make it clear whether a term is marked, we often attach explicitly
the mark $\sharp$ to meta variables for marked terms.
A \emph{constrained marked pair} over $(\cF\cup\Marked{\cD_\cR},\cG,\cPred)$ is a triple
$(\Marked{s},\Marked{t},\varphi)$, denoted by
$\CRule{\Marked{s}}{\Marked{t}}{\varphi}$, such that $s$ and $t$ are terms
in $T(\cF\cup\cG,\cV)$, both $s$ and $t$ are rooted by defined symbols
of $\cR$, and $\Var(s) \supseteq \Var(t) \cup \FVar(\varphi)$.
When $\varphi$ is $\top$, we may write $\Rule{\Marked{s}}{\Marked{t}}$
instead of $\CRule{\Marked{s}}{\Marked{t}}{\varphi}$. 
A constrained marked pair $\CRule{\Marked{s}}{\Marked{t}}{\varphi}$ is
called a \emph{dependency pair} 
of $\cR$ if there exists a renamed variant
$\CRule{s}{C[t]}{\varphi}$ of a rewrite rule in $\cR$.
We denote the set of dependency pairs of $\cR$ by $\DP(\cR)$. 
In the following, we let $\cS$ be a set of dependency pairs of $\cR$ unless noted otherwise. 

A (possibly infinite) derivation 
$\Marked{s_0}\sigma_0 \to_{\varepsilon,\cS} \Marked{t_0}\sigma_0
\to^*_{>\varepsilon,\cR}\Marked{s_1}\sigma_1 \to_{\varepsilon,\cS} \Marked{t_1}\sigma_1
\to^*_{>\varepsilon,\cR} \cdots$ 
with $\sigma_0,\sigma_1,\sigma_2,\ldots \in \Subst(\cF\cup\cG,\cV)$ 
is called a \emph{dependency chain w.r.t.\ $\cS$ and $\cR$} ($\DPproblem{\cS}{\cR}$-chain, for short).
The chain is called 
\emph{infinite} if it contains infinitely many $\to_{\varepsilon,\cS}$ steps, 
and called \emph{minimal} if $\Marked{t_i}\sigma_i$ is terminating
w.r.t.\ $\cR$ for all $i \geq 0$.
We deal with minimal chains only, and chains in this paper are minimal unless noted otherwise.
\begin{theorem}[\cite{SNS11}]
$\cR$ is terminating iff there is no
infinite 
$\DPproblem{\DP(\cR)}{\cR}$-chain.
\end{theorem}

A pair $(\cS,\cR)$ of sets of dependency pairs and constrained
rewrite rules is called a \emph{DP problem}.
We denote a DP problem $(\cS,\cR)$ by $\cS$ because in this paper, we do
not modify $\cR$.
A DP problem $\DPproblem{\cS}{\cR}$ is called \emph{finite} if there is
no infinite 
 $\DPproblem{\cS}{\cR}$-chain, and called \emph{infinite} if the DP problem
 is not finite or $\cR$ is not terminating. 
 Note that there are DP problems which are both finite and infinite
 (see~\cite{GTSF06}).
A DP problem $\DPproblem{\cS}{\cR}$ is called \emph{trivial} if
$\cS=\emptyset$.
A \emph{DP processor} is a function which
takes a DP problem as input and returns a finite set of DP problems. 
A DP processor $\Proc$ is called \emph{sound} if for any DP problem
$\DPproblem{\cS}{\cR}$, the DP problem is finite whenever all the DP problems in
$\Proc(\DPproblem{\cS}{\cR})$ are finite. 
$\Proc$ is called \emph{complete} if for any DP problem $\DPproblem{\cS}{\cR}$,
the DP problem is infinite whenever there exists an infinite DP problem in
$\Proc(\DPproblem{\cS}{\cR})$. 
The \emph{DP framework} is a method to
prove/disprove the finiteness of DP problems:%
\footnote{
In this paper, we do not consider disproving termination,
and thus, we do not formalize the case where DP processors return ``no''~\cite{GTSF06}.}
given a constrained TRS $\cR$, if the \emph{initial} DP problem
$\DPproblem{\DP(\cR)}{\cR}$ is decomposed into trivial DP problems by sound
DP processors, then the framework succeeds in proving termination of
$\cR$. 


In the rest of this section, we briefly introduce the DP processor based on 
\emph{polynomial interpretations} (PI, for short),
which is an extension of those in the DP framework for TRSs.

The \emph{PI-based} processor in~\cite{SNS11} is defined for constrained TRSs with
an LIA-structure $\cMint$ with binary predicate symbols $>$ and $\geq$.
Given a signature $\Sigma=\cF\uplus\cGint$ with $\cGint \supseteq \{+,-\}$, we define a \emph{linear polynomial interpretation}\/%
\footnote{
We consider ``linear'' polynomials only because we use PIs over the integers, which may have negative coefficients, and we interpret ground terms containing nests of defined symbols.}
 $\Pol$ \emph{for a subsignature $\cF' \subseteq \cF$}
via $\cGint$ as follows:
\begin{itemize}
  \item for any $n$-ary function symbol $f$ in $\cF'$, $\Pol(f)$ is a term in $T(\cGint,\{x_1,\ldots,x_n\})$ that represents a linear polynomial.
\end{itemize}
Note that 
$\cGint$ and $\cMint$ may be different from $\cGlia$ and $\cMlia$ in Example~\ref{ex:Glia}.
For readability, we use usual mathematical notions for terms in $T(\cGint,\cV)$, e.g., $100$ for $\symb{s}^\symb{100}(\symb{0})$, $2x$ for $x + x$, and so on.
In the following, given an $n$-ary symbol $f$ in $\cF'$, we write $a_0 + a_1x_1 + \cdots + a_n x_n$ for $\Pol(f)$
where $a_0,a_1,\ldots,a_n \in \Int$.
We apply $\Pol$ for $\cF'$ to arbitrary terms in $T(\cF\cup\cGint,\cV)$ as follows:
$\Pol(x) = x$ for $x \in \cV$;
$\Pol(f(t_1,\ldots,t_n)) = \Pol(f)\{ x_i \mapsto \Pol(t_i) \mid 1
       \leq i \leq n\}$ if $\Pol(f)$ is defined (i.e., $f \in \cF'$), 
and otherwise, 
$\Pol(f(t_1,\ldots,t_n)) = f(\Pol(t_1),\ldots,\Pol(t_n))$.
In the following, we use $\cR$ as a constrained TRS over $(\cF,\cGint,\cPint,\cMint)$ without notice. 
To simplify the presentation, we introduce a weaker version of the
PI-based processor in~\cite{SNS11}.
\begin{definition}[\cite{SNS11}]
 \label{def:PIProc}
Let $\Pol$ be a linear PI for
 $\Marked{\cD_\cR}$%
 \footnote{
 $\Pol$ is not defined for any symbols in $\cF$.} such that
\begin{itemize}
 \item[\rm\CondInterpreted]
 $\Pol(\Marked{s}),\Pol(\Marked{t})\in T(\cGint,\cV)$ for all $\CRule{\Marked{s}}{\Marked{t}}{\varphi}\in\cS$,%
 \item[\rm\CondPreserve] 
 $\Var(\Pol(\Marked{t})) \subseteq \FVar(\varphi)\cup\Var(\Pol(\Marked{s}))$ for all $\CRule{\Marked{s}}{\Marked{t}}{\varphi}\in \cS$,
 and
 \item[\rm\CondSdec] 
 $\varphi \Rightarrow \Pol(\Marked{s}) \geq \Pol(\Marked{t})$ is $\cMint$-valid for all $\CRule{\Marked{s}}{\Marked{t}}{\varphi}\in \cS$.
\end{itemize}
Then, the \emph{PI-based processor} $\PIProc$ is defined as follows:
\[
 \PIProc(\DPproblem{\cS}{\cR}) =
 \{~
   \DPproblem{\cS\setminus\Pgtr}{\cR},
   ~ \DPproblem{\cS\setminus\Pbound}{\cR},
   ~ \DPproblem{\cS\setminus\Pfilter}{\cR}
   ~\}
\]
where
\begin{itemize}
 \item $\Pgtr = \{ \CRule{\Marked{s}}{\Marked{t}}{\varphi} \in \cS \mid 
       \mbox{$\varphi \limplies \Pol(\Marked{s}) > \Pol(\Marked{t})$ is $\cMint$-valid}\,
       \}$,
 \item $\Pbound = \{ \CRule{\Marked{s}}{\Marked{t}}{\varphi} \in \cS \mid 
       \mbox{$\varphi \limplies \Pol(\Marked{s}) \geq c_0$ is $\cMint$-valid for some $c_0 \in T(\cGint)$}\,
       \}$,%
       \footnote{
       To simplify discussion, we consider a common ground term $c_0$ (the minimum one) such that $\varphi \limplies \Pol(\Marked{s}) > c_0$ is $\cMint$-valid for all $\CRule{\Marked{s}}{\Marked{t}}{\varphi} \in \Pbound$.
       }
       and
 \item $\Pfilter = \{ \CRule{\Marked{s}}{\Marked{t}}{\varphi} \in \cS \mid 
       \Var(\Pol(\Marked{s})) \subseteq \FVar(\varphi)\}$.
       \end{itemize}
\end{definition}
For a ground $\DPproblem{\cS}{\cR}$-chain, the assumptions {\CondInterpreted}, {\CondPreserve}, {\CondSdec} and the sets $\Pgtr$, $\Pbound$, $\Pfilter$ play the following role:
\begin{itemize}
	\item 
{\CondInterpreted} ensures that all the uninterpreted function symbols in
$\cS$ are dropped by applying $\Pol$ to pairs in $\cS$.
However, the application of $\Pol$ to an instance of a pair in $\cS$ may
contain an uninterpreted function symbol.
	\item
	Pairs in $\Pfilter$ ensure the existence of a pair which is reduced by $\Pol$ to an integer.
	\item
{\CondPreserve} 
ensures that all terms appeared after a rewrite step of $\Pfilter$ can be reduced by $\Pol$ to integers,
i.e., a suffix of the $\DPproblem{\cS}{\cR}$-chain can be converted by $\Pol$ to a sequence of integers.
	\item
{\CondSdec} ensures that the sequence of integers is decreasing.
	\item
Pairs in $\Pgtr$ ensure that the sequence obtained by taking integers corresponding to rewrite steps of $\Pgtr$ is strictly decreasing.
	\item 
Pairs in $\Pbound$ ensure the 
existence of a lower bound for the decreasing sequence if the $\DPproblem{\cS}{\cR}$-chain has infinitely many $\to_{\varepsilon,\Pbound}$-steps.
\end{itemize}
To make $\DPproblem{\cS}{\cR}$ smaller via $\PIProc$,
we need $\Pgtr\ne\emptyset$, $\Pbound\ne\emptyset$, and $\Pfilter \ne \emptyset$.
The idea of the PI-based processor in Definition~\ref{def:PIProc} is that
an infinite $\DPproblem{\cS}{\cR}$-chain which contains each pair in
$\Pgtr \cup \Pbound \cup \Pfilter$ infinitely many times 
can be transformed into an infinite bounded strictly-decreasing sequence of integers, while such a sequence does not exist.

\begin{theorem}[\cite{SNS11}]
\label{thm:soundness-completeness_of_PIproc}
$\PIProc$ is sound and complete.
\end{theorem}

\begin{example}
\label{ex:ack}
Consider the following constrained TRS defining Ackermann function over the integers, while $\symb{ack}$ is not defined for negative integers:
\[
 \cRack = 
 \left\{
 \begin{array}{r@{\>}c@{\>}l@{~~}c@{}c@{}c}
 \CRule{\symb{ack}(x,y) &}{& \symb{s}(y) &}{& x = \symb{0} \land y \geq \symb{0} &} \\
 \CRule{\symb{ack}(x,y) &}{& \symb{ack}(\symb{p}(x),\symb{s}(\symb{0})) &}{& x > \symb{0} \land y = \symb{0} &} \\
 \CRule{\symb{ack}(x,y) &}{& \symb{ack}(\symb{p}(x),\symb{ack}(x,\symb{p}(y))) &}{& x > \symb{0} \land y > \symb{0} &} \\
 \end{array}
 \right\}
 \cup
 \cRsp
\]
The following are the dependency pairs of $\cRack$:
\[
 \DP(\cRack) = 
 \left\{
 \begin{array}{rr@{\>}c@{\>}l@{~~}c@{}c@{}c}
 (7) & \CRule{\symb{ack}(x,y) &}{& \Marked{\symb{s}}(y) &}{& x = \symb{0} \land y \geq \symb{0} &} \\
 (8) & \CRule{\Marked{\symb{ack}}(x,y) &}{& \Marked{\symb{ack}}(\symb{p}(x),\symb{s}(\symb{0})) &}{& x > \symb{0} \land y = \symb{0} &} \\
 (9) & \CRule{\Marked{\symb{ack}}(x,y) &}{& \Marked{\symb{p}}(x) &}{& x > \symb{0} \land y = \symb{0} &} \\
 (10) & \CRule{\Marked{\symb{ack}}(x,y) &}{& \Marked{\symb{s}}(\symb{0}) &}{& x > \symb{0} \land y = \symb{0} &} \\
 (11) & \CRule{\Marked{\symb{ack}}(x,y) &}{& \Marked{\symb{ack}}(\symb{p}(x),\symb{ack}(x,\symb{p}(y))) &}{& x > \symb{0} \land y > \symb{0} &} \\
 (12) & \CRule{\Marked{\symb{ack}}(x,y) &}{& \Marked{\symb{p}}(x) &}{& x > \symb{0} \land y > \symb{0} &} \\
 (13) & \CRule{\Marked{\symb{ack}}(x,y) &}{& \Marked{\symb{ack}}(x,\symb{p}(y)) &}{& x > \symb{0} \land y > \symb{0} &} \\
 (14) & \CRule{\Marked{\symb{ack}}(x,y) &}{& \Marked{\symb{p}}(y) &}{& x > \symb{0} \land y > \symb{0} &} \\
 \end{array}
 \right\}
\]
By using the DP processor based on \emph{strongly connected components} (cf.~\cite{SNS11}), we can drop (7), (9), (10), (12), and (14) from the initial DP problem $\DPproblem{\DP(\cRack)}{\cRack}$, obtaining the DP problem $\DPproblem{\{(8),(11),(13)\}}{\cRack}$.
Let us try to prove finiteness of the DP problem $\DPproblem{\{(8),(11),(13)\}}{\cRack}$.
Let $\Pol$ be a linear PI such that $\Pol(\Marked{\symb{ack}}) = x_1$.
Then, the assumptions {\CondInterpreted}, {\CondPreserve}, and {\CondSdec} in Definition~\ref{def:PIProc} are satisfied, and we have that
$\Pgtr = \{ (8), (11) \}$ and $\Pbound=\Pfilter=\{(8),(11),(13)\}$.
Thus, $\PIProc(\DPproblem{\{(8),(11),(13)\}}{\cRack})=\{ \DPproblem{\{(13)\}}{\cRack}, \DPproblem{\emptyset}{\cRack} \}$.
Let $\Pol'$ be a linear PI such that $\Pol(\Marked{\symb{ack}}) = x_2$.
Then, the assumptions {\CondInterpreted}, {\CondPreserve}, and {\CondSdec} in Definition~\ref{def:PIProc} are satisfied, and we have that
$\Pgtr = \Pbound=\Pfilter=\{(13)\}$.
Thus, $\PIProc(\DPproblem{\{(13)\}}{\cRack}) = \{ \DPproblem{\emptyset}{\cRack} \}$.
Therefore, $\cRack$ is terminating.
\end{example}

\begin{example}
\label{ex:PI}
Consider $\cRMc$ and its dependency pairs $\DP(\cRMc)$ in Section~\ref{sec:intro} again.
By using the DP processor based on strongly connected components, we can drop (5) and (6) from the initial DP problem $\DPproblem{\DP(\cRMc)}{\cRMc}$, obtaining the DP problem $\DPproblem{\{(3),(4)\}}{\cRMc}$.
Let us try to apply $\PIProc$ to the DP problem $\DPproblem{\{(3),(4)\}}{\cRMc}$.
Let $\Pol$ be a linear PI such that $\Pol(\Marked{\symb{f}})= a_0 + a_1x_1$.
To satisfy {\CondInterpreted},
$a_1$ has to be $0$ since $\symb{f} \notin \Marked{\cD_\cR}$.
Thus, 
 $\Pol(\Marked{\symb{f}})= a_0$,
and hence $\Pgtr = \emptyset$.
Therefore,
$\PIProc(\DPproblem{\{(3),(4)\}}{\cRMc}) =
\{~
\DPproblem{\{(3),(4)\}}{\cRMc}
~\}
$ 
and $\PIProc$ does not work for the DP problem $\DPproblem{\DP(\cRMc)}{\cRMc}$.
Note that the other DP processors based on strongly connected components or \emph{the subterm criterion} (cf.~\cite{SNS11}) do not work for this DP problem, either. 
\end{example}

\section{From Dependency Chains to Monotone Sequences of Integers}
\label{sec:improvement}

PIs satisfying the conditions in Definition~\ref{def:PIProc} transform $\DPproblem{\cS}{\cR}$-chains into bounded decreasing sequences of integers.
Focusing on such PIs, we obtain the following corollary from Definition~\ref{def:PIProc} and Theorem~\ref{thm:soundness-completeness_of_PIproc}.
\begin{corollary}
Let $\Pol$ be a linear PI\/ for
 $\Marked{\cD_\cR}$ such that
     {\rm\CondInterpreted}, {\rm\CondPreserve}, and {\rm\CondSdec} in Definition~\ref{def:PIProc} hold.
Then, every ground $\DPproblem{\cS}{\cR}$-chain $\Marked{s_0}\sigma_0 \to_{\varepsilon,\cS} \Marked{t_0}\sigma_0
\to^*_{>\varepsilon,\cR}\Marked{s_1}\sigma_1 \to_{\varepsilon,\cS} \Marked{t_1}\sigma_1
\to^*_{>\varepsilon,\cR} \cdots$ starting with $\CRule{\Marked{s_0}}{\Marked{t_0}}{\varphi_0}$ satisfying $\Var(\Pol(\Marked{s_0})) \subseteq \Var(\varphi_0)$
can be transformed into a decreasing sequence $\Pol(\Marked{s_0}\sigma_0) \geq \Pol(\Marked{t_0}\sigma_0) \geq \Pol(\Marked{s_1}\sigma_1) \geq \Pol(\Marked{t_1}\sigma_1) \geq \cdots$ of integers such that
\begin{itemize}
	\item $>$ appears infinitely many times if $\CRule{\Marked{s}}{\Marked{t}}{\varphi} \in \Pgtr$ in Definition~\ref{def:PIProc} appears in the $\DPproblem{\cS}{\cR}$-chain infinitely many times, 
		and
	\item the sequence is bounded (i.e., there exists an integer $n$ such that $\Pol(\Marked{s_i}) \geq n$ for all $i$) if $\CRule{\Marked{s}}{\Marked{t}}{\varphi} \in \Pbound$ in Definition~\ref{def:PIProc} appears in the $\DPproblem{\cS}{\cR}$-chain infinitely many times.
\end{itemize}
\end{corollary}
To show the non-existence of infinite $\DPproblem{\cS}{\cR}$-chains, it suffices to show the non-existence of infinite ground $\DPproblem{\cS}{\cR}$-chains.
This is because the signature contains an interpreted constant, e.g., $\symb{0}$, and we can make any $\DPproblem{\cS}{\cR}$-chain ground by instantiating the $\DPproblem{\cS}{\cR}$-chain with an interpreted constant.

In this section, we show sufficient conditions of a linear PI for transforming dependency chains into  monotone sequences of integers, strengthening the PI-based processor $\PIProc$.
The difference from $\PIProc$ is to take $\cR$ into account.

\subsection{The Existing Approach to Transformation of Chains into Decreasing Sequences}
\label{subsec:dec-R}

As the first step, we follow the existing approach in~\cite{FK08}.
To this end, we recall the notion of \emph{reducible positions}~\cite{FK08}. 
A natural number $i$ is a \emph{reducible position} of a marked symbol $\Marked{f}$ w.r.t.\ $\cS$ if there is a dependency pair $\CRule{\Marked{s}}{\Marked{f}(t_1,\ldots,t_n)}{\varphi} \in \cS$ such that $t_i \notin T(\cG,\Var(\varphi))$.%
\footnote{
In~\cite{FK08}, ``$t_i \notin T(\cG,\cV)$'' is required but in this paper, we require a stronger condition ``$t_i\notin T(\cG,\Var(\varphi))$'' that is more essential for this notion.}

To extract rewrite sequences of $\cR$ in transforming chains into sequences of integers, 
for an $n$-ary symbol $f$ with $\Pol(\Marked{f})=a_0+a_1x_1+\cdots + a_nx_n$, $\PIProc$ requires {\CondInterpreted}---the coefficient $a_i$ of any reducible position $i$ of $\Marked{f}$ w.r.t.\ $\cS$ to be $0$.
Due to this requirement, in applying $\PIProc$, we do not have to take into account rules in $\cR$.
However, as seen in Example~\ref{ex:PI}, this requirement makes $\PIProc$ ineffective
in the case where all arguments of marked symbols are reducible positions.
For this reason, we relax this requirement as in~\cite{FK08} by making a linear PI $\Pol$ for $\Marked{\cD_\cR}\cup\cF$ satisfy {\CondSdec} and the following conditions:
\begin{itemize}
\item[\CondSubtraction] Any reduction of $\cR$ for uninterpreted symbols in $\cF$ does not happen in the second argument of the subtraction operator, i.e., for any $\CRule{u}{v}{\varphi}\in \cR \cup \cS$ and any subterm $v'$ of $v$, if $v'$ is rooted by the subtraction symbol ``$-$'', then $v'|_2 \in T(\cGint,\FVar(\varphi))$;
\item[\CondPolF] $b_1,\ldots,b_n \geq 0$ for all $f/n \in \cF$ with $\Pol(f) = b_0 + b_1x_1 + \cdots + b_nx_n$;
\item[\CondRdec] $\varphi \Rightarrow \Pol(\ell) \geq \Pol(r)$ is $\cMint$-valid for all $\CRule{\ell}{r}{\varphi}\in \cR$;
\item[\CondPolMdec] $a_i \geq 0$ for any reducible position $i$ of any $\Marked{f}/n$ with $\Pol(\Marked{f}) = a_0 + a_1x_1 + \cdots + a_nx_n$.
\end{itemize}
The first three conditions ensure that for any term $s,t \in T(\cF\cup\cG)$, if $s \mathrel{\to_\cR} t$, then $(\Pol(s))^{\cMint} \geq (\Pol(t))^{\cMint}$ holds.
In addition to the first three conditions, the last condition ensures that for any term $s,t \in T(\cF\cup\cG)$ with $\Root(s) \in \cD_\cR$, if $\Marked{s} \mathrel{\to_\cR} \Marked{t}$, then $(\Pol(\Marked{s}))^{\cMint} \geq (\Pol(\Marked{t}))^{\cMint}$ holds.
Note that {\CondPreserve} and $\Pfilter$ are no longer required because $\Pol$ for $\Marked{\cD_\cR}\cup\cF$ interprets all uninterpreted symbols.

\begin{example}
\label{ex:CondRdec-impossible}
Let $\Pol$ be a linear PI such that $\Pol(\Marked{\symb{f}})= a_0 + a_1x_1$ and $\Pol(\symb{f})= b_0 + b_1x_1$ with $a_1 \geq 0$ and $b_1 \geq 0$.
To transform $\DPproblem{\{(3),(4)\}}{\cRMc}$-chains into decreasing sequences of integers, both $\symb{s}^{101}(\symb{0}) > x \Rightarrow \Pol(\symb{f}(x)) \geq \Pol(\symb{f}(\symb{f}(\symb{s}^{11}(x))))$ (i.e., $101 > x \Rightarrow b_0 + b_1 x \geq b_0 + b_1(b_0 + b_1(x + 11))$) and 
$\neg (\symb{s}^{101}(\symb{0}) > x) \Rightarrow \Pol(\symb{f}(x)) \geq \Pol(\symb{p}^{10}(x))$ (i.e., $101 \leq x \Rightarrow b_0 + b_1 x \geq x - 10$) have to be $\cMlia$-valid.
However, there is no assignment for $a_0,a_1,b_0,b_1$ ensuing the validity of the two formulas.
\end{example}

\subsection{Transforming Rewrite Sequences into Increasing Sequences of Integers}
\label{subsec:inc-R}

To preserve monotonicity of linear PIs, we keep the assumption {\CondPolF}.
Under {\CondPolF},
as seen in Example~\ref{ex:CondRdec-impossible}, it is impossible for any linear PI to ensure 
{\CondRdec} for $\cRMc$.
Then, let us try to transform ground rewrite sequences of $\cR$ into \emph{increasing} sequences of integers. 
This is a key idea of improving $\PIProc$.
To transform ground rewrite sequences of $\cR$ into increasing sequences, we require the following condition instead of {\CondRdec}:
\begin{itemize}
\item[\CondRinc] $\varphi \Rightarrow \Pol(\ell) \mathrel{\textcolor{blue}{\leq}} \Pol(r)$ is $\cMint$-valid for any $\CRule{\ell}{r}{\varphi}\in \cR$.
\end{itemize}
When transforming both ground dependency chains and ground rewrite sequences of $\cR$ into decreasing sequences, the coefficient for a reducible position (i.e., $a_i$ of $\Pol(\Marked{f})=a_0+a_1x_1 + \cdots + a_nx_n$ with reducible position $i$ of $\Marked{f}$) has to be a non-negative integer because 
any rewrite sequences appears below the reducible position is transformed into a decreasing sequence.
On the other hand, when transforming rewrite sequences of $\cR$ into increasing sequences, 
all coefficients for reducible positions have to be non-positive.
Thus, we modify the assumption {\CondPolMdec} as follows:
\begin{itemize}
\item[\CondPolMinc] $a_i \mathrel{\textcolor{blue}{\leq}} 0$ for any reducible position $i$ of any $\Marked{f}/n$ with $\Pol(\Marked{f}) = a_0 + a_1x_1 + \cdots + a_nx_n$.
\end{itemize}
Under the assumptions {\CondSdec}, {\CondSubtraction}, {\CondPolF}, {\CondRinc}, and {\CondPolMinc}, any ground $\DPproblem{\cS}{\cR}$-chain is transformed into a decreasing sequence of integers.

\begin{example}
Let $\Pol$ be a linear PI such that $\Pol(\Marked{\symb{f}})= -1 - x_1$ and $\Pol(\symb{f})= -10 + x_1$.
Then, all {\CondSdec}, {\CondSubtraction}, {\CondPolF}, {\CondRinc}, and {\CondPolMinc} are satisfied, and $\Pgtr=\Pbound=\{~ (3), ~ (4) ~\}$.
This means that every ground $\DPproblem{\{(3),(4)\}}{\cRMc}$-chain 
can be transformed into a decreasing sequence of integers such that
$>$ appears infinitely many times and
the sequence is bounded.
If each of $(3)$ and $(4)$ appears in a ground $\DPproblem{\{(3),(4)\}}{\cRMc}$-chain infinitely many times, then the $\DPproblem{\{(3),(4)\}}{\cRMc}$-chain is transformed into a bounded strictly-decreasing sequence of integers, but such a sequence does not exists.
This means that $(3)$ and $(4)$ appears in any ground $\DPproblem{\{(3),(4)\}}{\cRMc}$-chain finitely many times.
Therefore, there is no infinite ground $\DPproblem{\{(3),(4)\}}{\cRMc}$-chain, and hence $\cRMc$ is terminating.
\end{example}

\subsection{Transforming Dependency Chains into Increasing Sequences of Integers}
\label{subsec:inc-DP}

The role of PI $\Pol$ in $\PIProc$ is to transform dependency chains into bounded \emph{decreasing} sequences of integers, and to drop a dependency pair 
$\CRule{\Marked{s}}{\Marked{t}}{\varphi} \in \cS$ such that $\varphi \limplies \Pol(\Marked{s}) > \Pol(\Marked{t})$ is $\cMint$-valid.
Since transformed sequences are bounded, the sequences do not have to be decreasing, i.e., they may be bounded \emph{increasing} sequences. 
To transform dependency chains into increasing sequences, we invert $\geq$ in {\CondSdec} and $>$ of $\Pgtr$ as follows:
	\begin{itemize}
\item[\CondSinc] $\varphi \Rightarrow \Pol(\Marked{s}) \mathrel{\textcolor{blue}{\leq}} \Pol(\Marked{t})$ is valid for all $\CRule{\Marked{s}}{\Marked{t}}{\varphi}\in \cS$, and
	\item $\Pgtr = \{ \CRule{\Marked{s}}{\Marked{t}}{\varphi} \in \cS \mid 
       \mbox{$\varphi \limplies \Pol(\Marked{s}) \mathrel{\textcolor{blue}{<}} \Pol(\Marked{t})$ is $\cMint$-valid}\,
       \}$.
\end{itemize}
For ground rewrite sequences, we have two ways to transform them (into either decreasing or increasing sequences) and thus, we have the following two combinations to transform dependency chains into increasing sequences:
in addition to {\CondSubtraction}, {\CondPolF}, and {\CondSinc},
\begin{itemize}
	\item When transforming ground rewrite sequences into decreasing sequences as in Section~\ref{subsec:dec-R}, we require {\CondRdec} and {\CondPolMinc}.
	
	\item When transforming ground rewrite sequences into increasing sequences as in Section~\ref{subsec:inc-R}, we require {\CondRinc} and {\CondPolMdec}.
\end{itemize}
For the both cases above, to ensure the existence of an upper bound, we need a dependency pair $\CRule{\Marked{s}}{\Marked{t}}{\varphi}\in \cS$ such that $\varphi \Rightarrow \Pol(\Marked{s}) \mathrel{\textcolor{blue}{\leq}} c_0$ is $\cMint$-valid for some $c_0 \in T(\cGint)$.

\subsection{Improving the PI-based Processor}

Finally, we formalize the ideas in previous sections as an improvement of the PI-based processor $\PIProc$.
\begin{definition}
 \label{def:newPIProc}
Let $(\rel_1,\rel_2,\rel_3) \!\in\! \{\DecDec,\IncDec,\DecInc, \IncInc\}$,
and suppose that {\rm\CondSubtraction} 
holds.
Let $\Pol$ be a linear PI for
 $\Marked{\cD_\cR}\cup\cF$ such that
\begin{itemize}
\item $b_i\geq 0$ for all $1 \leq i \leq n$ and for any $f/n \in \cF$ with $\Pol(f) = b_0 + b_1x_1 + \cdots + b_nx_n$, 
\item $\varphi \Rightarrow \Pol(\ell) \mathrel{\rel_2} \Pol(r)$ is $\cMint$-valid for all $\CRule{\ell}{r}{\varphi}\in \cR$, 
\item $a_i \mathrel{\rel_3} 0$ for any reducible position $i$ of any $\Marked{f}/n$ and $\Pol(\Marked{f}) = a_0 + a_1x_1 + \cdots + a_nx_n$, 
	and
 \item  
$\varphi \limplies (\Pol(\Marked{s}) \mathrel{\rel_1} \Pol(\Marked{t}) \lor \Pol(\Marked{s}) \simeq \Pol(\Marked{t}))$ is $\cMint$-valid
for all $\CRule{\Marked{s}}{\Marked{t}}{\varphi} \in \cS$. 
\end{itemize}
Then, the \emph{PI-based processor} $\PIProcX$ is defined as follows:
\[
 \PIProcX(\DPproblem{\cS}{\cR}) =
 \{~
   \DPproblem{\cS\setminus\newPgtr}{\cR},
   ~ \DPproblem{\cS\setminus\Pbound}{\cR}
   ~\}
\]
where 
\begin{itemize}
 \item $\newPgtr = \{ \CRule{\Marked{s}}{\Marked{t}}{\varphi} \in \cS \mid 
       \mbox{$\varphi \limplies \Pol(\Marked{s}) \mathrel{\rel_1} \Pol(\Marked{t})$ is $\cMint$-valid}\,
       \}$, and
 \item $\Pbound = \{ \CRule{\Marked{s}}{\Marked{t}}{\varphi} \in \cS \mid 
       \mbox{$\varphi \limplies \Pol(\Marked{s}) \mathrel{\rel_1} c_\symb{0} \lor \Pol(\Marked{s}) \simeq c_\symb{0}$ is $\cMint$-valid for some $c_\symb{0} \in T(\cGint)$}\,
       \}$.
       \end{itemize}
\end{definition}

Before proving soundness and completeness of $\PIProcX$, we show some properties of ground dependency chains and ground rewrite sequences w.r.t.\ {\CondRdec}, {\CondRinc}, {\CondSdec}, {\CondSinc}, etc.
The following lemmas hold by assumptions. 
\begin{lemma}
Let $s,t$ be ground terms in $T(\cF\cup\cGint)$, and $\Pol$ a linear PI for $\Marked{\cD_\cR}\cup\cF$ such that {\rm\CondSubtraction} and {\rm\CondPolF} hold.
Suppose that $s \mathrel{\to_{p,\cR}} t$ and $p$ is not a position below the second argument of ``$-$''.	
	\begin{itemize}
	\item If\/ {\rm\CondRdec} holds, then $(\Pol(s))^{\cMint} \geq (\Pol(t))^{\cMint}$ holds.
	\item If\/ {\rm\CondRinc} holds, then $(\Pol(s))^{\cMint} \leq (\Pol(t))^{\cMint}$ holds.
	\end{itemize}
\end{lemma}
%
\begin{lemma}
\label{lem:monotonicity}
Let $s,t$ be ground terms in $T(\cF\cup\cGint)$, and $\Pol$ a linear PI for $\Marked{\cD_\cR}\cup\cF$ such that {\rm\CondSubtraction} and {\rm\CondPolF} hold.
Suppose that $s,t$ are rooted by $f/n \in \cD_\cR$.
	\begin{itemize}
		\item If $\Marked{s} \mathrel{\to_{i.p,\cR}} \Marked{t}$ and $s|_i \in T(\cGint)$ for some $i \in \{1,\ldots,n\}$ and some position $p$, then $(\Pol(\Marked{s}))^{\cMint} = (\Pol(\Marked{t}))^{\cMint}$ holds.
		\item If\/ {\rm\CondSdec} holds and $\Marked{s} \mathrel{\to_{\cS}} \Marked{t}$, then $(\Pol(\Marked{s}))^{\cMint} \geq (\Pol(\Marked{t}))^{\cMint}$ holds.
 	\item If\/ {\rm\CondSinc} holds and $\Marked{s} \mathrel{\to_{\cS}} \Marked{t}$, then $(\Pol(\Marked{s}))^{\cMint} \leq (\Pol(\Marked{t}))^{\cMint}$ holds.
	\item Suppose that 
	$\Marked{s_0} \mathrel{\to_{i_1.p,\cR}} \Marked{s_2} \mathrel{\to_{i_2.p,\cR}} \cdots \mathrel{\to_{i_n.p,\cR}} \Marked{s_n}$ where $i_1,\ldots,i_n$ are reducible positions of $\Marked{f}$.
	\begin{itemize}
		\item If\/ {\rm\CondRdec} and\/ {\rm\CondPolMdec} hold or {\rm\CondRinc} and\/ {\rm\CondPolMinc} hold, then
 $(\Pol(\Marked{s_0}))^{\cMint} \geq (\Pol(\Marked{s_1}))^{\cMint} \geq \cdots \geq (\Pol(\Marked{s_n}))^{\cMint}$ holds.
		\item If\/ {\rm\CondRdec} and\/ {\rm\CondPolMinc} hold or {\rm\CondRinc} and\/ {\rm\CondPolMdec} hold, then
 $(\Pol(\Marked{s_0}))^{\cMint} \leq (\Pol(\Marked{s_1}))^{\cMint} \leq \cdots \leq (\Pol(\Marked{s_n}))^{\cMint}$ holds.
 	\end{itemize}
	\end{itemize}
\end{lemma}

In addition to the above lemmas, we introduce a key lemma that makes the proof of soundness routine.
Let $\cS' \subseteq \cS$.
An infinite $\DPproblem{\cS}{\cR}$-chain is called \emph{$\cS'$-innumerable} if every element in $\cS'$ appears in the chain infinitely many times~\cite{SNS11}.
\begin{lemma}[\cite{SNS11}]
\label{lem:S-innumerable}
Let a DP processor $\Proc$ such that for any DP problem $\DPproblem{\cS}{\cR}$, $\Proc(\DPproblem{\cS}{\cR}) \subseteq 2^{\cS}$.
Then, $\Proc$ is sound and complete if for any DP problem $\DPproblem{\cS}{\cR}$, there exists no $\cS'$-innumerable chain for any $\cS' \subseteq \cS$ such that $\cS' \setminus \cS'' \ne \emptyset$ for all $\cS'' \in \Proc(\DPproblem{\cS}{\cR})$.
\end{lemma}

Finally, we show soundness and completeness.
\begin{theorem}
\label{thm:soundness-completeness_of_newPIproc}
$\PIProcX[\DecDec]$, $\PIProcX[\IncDec]$, $\PIProcX[\DecInc]$, and $\PIProcX[\IncInc]$ are sound and complete.
\end{theorem}
\begin{proof}
We only consider the case of $\PIProcX[\DecInc]$.
The proofs of the remaining cases analogous.
The proof below follows that of~\cite[Theorem~3.3]{SNS11}.
The only difference from those proofs is the treatment of $\Marked{s} \mathrel{\to^*_\cR} \Marked{t}$.
By Lemma~\ref{lem:S-innumerable}, it suffices to show that for any subset $\cS' \subseteq \cS$ with $\cS'\cap \newPgtr \ne \emptyset$ and $\cS' \cap \Pbound \ne \emptyset$, there is no $\cS'$-innumerable $\DPproblem{\cS}{\cR}$-chain.
We proceed by contradiction.
Suppose that there exists some subset $\cS' \subseteq \cS$ such that $\cS'\cap \newPgtr \ne \emptyset$, $\cS' \cap \Pbound \ne \emptyset$, and there exists an $\cS'$-innumerable $\DPproblem{\cS}{\cR}$-chain.
Then, we can assume w.l.o.g.\ that the $\cS'$-innumerable chain is ground.%
\footnote{
If the infinite chain is not ground, then we can instantiate it with ground terms e.g., $\symb{0}$.
The existence of interpreted ground terms is ensured by the user-specified structure (see the definition of structures).}
Let $\Marked{s_0} \to_{\varepsilon,\cS} \Marked{t_0}
\to^*_{>\varepsilon,\cR}\Marked{s_1} \to_{\varepsilon,\cS} \Marked{t_1}
\to^*_{>\varepsilon,\cR} \cdots$ 
be the $\cS'$-innumerable ground $\DPproblem{\cS}{\cR}$-chain.
It follows from Lemma~\ref{lem:monotonicity} that for all $i \geq 0$, 
\[
(\Pol(\Marked{s_i}))^{\cMint} \leq (\Pol(\Marked{t_i}))^{\cMint}
~~\mbox{and}~~
(\Pol(\Marked{t_i}))^{\cMint} \leq (\Pol(\Marked{s_{i+1}}))^{\cMint}.
\]
There exists a ground term $c_0$ (an upper bound) such that $(\Pol(\Marked{s_i}))^{\cMint} \leq (\Pol(c_0))^{\cMint}$ holds for all $i$ with $\Marked{s_i} \mathrel{\to_{\varepsilon,\Pbound}} \Marked{t_i}$.
Since the chain is $\cS'$-innumerable and $\cS'\cap\Pbound\ne\emptyset$, $\Pbound$-steps appears in the chain infinitely many times, and thus, $(\Pol(\Marked{s_i}))^{\cMint} \leq (\Pol(c_0))^{\cMint}$ holds for all $i \geq 0$.
In addition, we have that $(\Pol(\Marked{s_i}))^{\cMint} < (\Pol(t_i))^{\cMint}$ holds for all $i$ such that $\Marked{s_i} \mathrel{\to_{\varepsilon,\newPgtr}} \Marked{t_i}$.
It follows from the assumptions (the chain is $\cS'$-innumerable, $\cS'\cap\newPgtr\ne\emptyset$, and $\cS'\cap\Pbound\ne\emptyset$) that $\newPgtr$-steps appears in the chain infinitely many times, and thus, the increasing sequence $(\Pol(\Marked{s_0}))^{\cMint} \leq (\Pol(\Marked{t_0}))^{\cMint} \leq (\Pol(\Marked{s_1}))^{\cMint} \leq (\Pol(\Marked{t_1}))^{\cMint} \leq \cdots$ contains infinitely many strictly increasing steps ($<$) while all elements are less than or equal to $(\Pol(c_0))^{\cMint}$.
This contradicts the fact that there is no bounded strictly-increasing infinite sequence of integers.
\qed	
\end{proof}

By definition, it is clear that $\PIProcX[\DecDec]$ and $\PIProcX[\IncDec]$ are the same functions from theoretical point of view, and $\PIProcX[\DecInc]$ and $\PIProcX[\IncInc]$ are so.
For example, given a DP problem $\DPproblem{\cS}{\cR}$ and a linear PI $\Pol_1$ satisfying the conditions of $\PIProcX[\DecDec]$, we can construct a linear PI $\Pol_2$ such that $\Pol_2$ satisfies the conditions of $\PIProcX[\IncDec]$ and $\PIProcX[\DecDec](\DPproblem{\cS}{\cR})=\PIProcX[\IncDec](\DPproblem{\cS}{\cR})$.

\subsection{Implementation and Experiments}
\label{subsec:execution}

We implemented the new PI-based processors $\PIProcX[\DecDec]$, $\PIProcX[\IncDec]$, $\PIProcX[\DecInc]$, and $\PIProcX[\IncInc]$ in \textsf{Cter}, a termination prover based on the techniques in~\cite{SNS11}.
Those processors first generate a template of a linear PI such as $\Pol(f)=a_0 + a_1 x_1 + \cdots a_n x_n$ with non-fixed coefficients $a_0,a_1,\ldots,a_n$, producing a non-linear integer arithmetic formula that belongs to \texttt{NIA}, a logic category of \textsf{SMT-LIB}%
\footnote{
\url{http://smtlib.cs.uiowa.edu}}.
Satisfiability of the generated formula corresponds to the existence of the PI satisfying the conditions that the processors require.
Then, the processors call \textsf{Z3}~\cite{z3}, an SMT solver, to find an expected PI:
if \textsf{Z3} returns ``\verb|unsat|'', then there exists no PI satisfying the conditions.

Table~\ref{tbl:runtime} illustrates the results of experiments to prove termination of $\cRMc$ using one of the new processors with timeout (3,600 seconds).
Experiments are conducted on a machine running Ubuntu 14.04 LTS equipped with an Intel Core i5 CPU at 3.20 GHz with 8 GB RAM.
$\PIProcX[\DecDec]$, $\PIProcX[\IncDec]$, and $\PIProcX[\DecInc]$ were applied once, but 
$\PIProcX[\IncInc]$ was applied twice---it first decomposes the DP problem $\DPproblem{\{(3),(4)\}}{\cRm}$ to $\DPproblem{\{(3)\}}{\cRm}$ and then solves $\DPproblem{\{(3)\}}{\cRm}$.
This means that \textsf{Z3} is called once or twice to check satisfiability of a formula given by the processor.
To show how PI-based processors work, Table~\ref{tbl:runtime} shows both the results of \textsf{Z3} and the corresponding PIs if existing.
Surprisingly, for $\PIProcX[\DecInc]$ and $\PIProcX[\IncInc])$ with the same power, the execution times are quite different.
The difference might be caused by how \textsf{Z3} searches assignments that satisfy given formulas.

\begin{table}[t]
\caption{the result of experiments to prove termination of $\cRMc$ using the new PI-based processors}
\label{tbl:runtime}
\small 
\centering
\begin{tabular}{|@{\,}c@{\,}|@{~}c@{~}|@{\,}c@{\,}||@{~}c@{~}|@{\,}r@{\,}||@{\,}c@{\,}|@{\,}r@{\,}|@{\,}r@{\,}|@{\,}r@{\,}|}
\hline 
Used processor 
& Chains & Rewrite sequences & Result & Time (sec.) &  Output of \textsf{Z3} & \multicolumn{1}{@{\,}c@{\,}|@{\,}}{$\Pol(\symb{f})$} & \multicolumn{1}{@{\,}c@{\,}|@{\,}}{$\Pol(\Marked{\symb{f}})$} & \multicolumn{1}{c|}{$c_0$} \\ 
\hline \hline
$\PIProcX[\DecDec]$ & decreasing & decreasing & failure & 0.60 &  unsat & \multicolumn{1}{@{\,}c@{\,}|@{\,}}{---} & \multicolumn{1}{@{\,}c@{\,}|@{\,}}{---} & \multicolumn{1}{c|}{--} \\ 
\hline 
$\PIProcX[\IncDec]$ & increasing & decreasing & failure & 0.56 & unsat & \multicolumn{1}{@{\,}c@{\,}|@{\,}}{---} & \multicolumn{1}{@{\,}c@{\,}|@{\,}}{---} & \multicolumn{1}{c|}{--} \\ 
\hline 
$\PIProcX[\DecInc]$ & decreasing &increasing & {\Success} & 30 &
sat &
$-10+x_1$ &
$-1-x_1$ &
$-101$
\\ 
\hline 
$\PIProcX[\IncInc]$ & increasing & increasing & {\Success} & 439 & \begin{tabular}{@{}c@{}}sat\\sat\end{tabular} & \multicolumn{1}{@{\,}c@{\,}|@{\,}}{\begin{tabular}{@{}c@{}}$-11+x_1$\\$-10+x_1$\end{tabular}} & \multicolumn{1}{@{\,}c@{\,}|@{\,}}{\begin{tabular}{@{}r@{}}$x_1$\\$~~2+x_1$\end{tabular}} & \multicolumn{1}{c|}{\begin{tabular}{@{}r@{}}$100$\\$200$\end{tabular}}  \\ 
\hline 
\end{tabular} 
\end{table}

Table~\ref{tbl:comparison} shows the results (``success'', ``failure'', or ``timeout'' by 3,600 seconds, with execution time) of proving termination of the following examples with nested recursions over the integers by using 
\textsf{Cter} with our previous or new PI-based processors,
{\AProVE}~\cite{AProVE} that proves termination of ITRSs~\cite{FGPSF09}, and 
\textsf{Ctrl}~\cite{KN15lpar} that proves termination of LCTRSs~\cite{KN13frocos}:
\begin{itemize}
	\item a variant of $\cRMc$
\[
 \cRMc' = 
 \left\{
 \begin{array}{r@{\>}c@{\>}l@{~~}c@{}c@{}c}
 \CRule{\symb{f}(x) &}{& \symb{f}(\symb{f}(\symb{s}^2(x))) &}{& \symb{s}^4(\symb{0}) > x &} \\
 \CRule{\symb{f}(x) &}{& \symb{p}(x) &}{& \neg (\symb{s}^4(\symb{0}) > x) &} \\
 \end{array}
 \right\}
 \cup
 \cRsp
\]
	\item a variant of $\cRack$ where $\symb{ack}$ is totally defined for the integers
\[
 \cRack' = 
 \left\{
 \begin{array}{r@{\>}c@{\>}l@{~~}c@{}c@{}c}
 \CRule{\symb{ack}(x,y) &}{& \symb{s}(y) &}{& x \leq \symb{0} &} \\
 \CRule{\symb{ack}(x,y) &}{& \symb{ack}(\symb{p}(x),\symb{s}(\symb{0})) &}{& x > \symb{0} \land y \leq \symb{0} &} \\
 \CRule{\symb{ack}(x,y) &}{& \symb{ack}(\symb{p}(x),\symb{ack}(x,\symb{p}(y))) &}{& x > \symb{0} \land y > \symb{0} &} \\
 \end{array}
 \right\}
 \cup
 \cRsp
\]
	\item $\mathsf{nest}$ in~\cite{Giesl97}
\[
 \cRnest = 
 \left\{
 \begin{array}{r@{\>}c@{\>}l@{~~}c@{}c@{}c}
 \CRule{\symb{nest}(x) &}{& \symb{0} &}{& x \leq \symb{0} &} \\
 \CRule{\symb{nest}(x) &}{& \symb{nest}(\symb{nest}(\symb{p}(x))) &}{& x > \symb{0} &} \\
 \end{array}
 \right\}
 \cup
 \cRsp
\]
	\item a variant of $\cRnest$
\[
 \cRnest' = 
 \left\{
 \begin{array}{r@{\>}c@{\>}l@{~~}c@{}c@{}c}
 \CRule{\symb{nest}(x) &}{& \symb{s}^3(\symb{0}) &}{& x \leq \symb{s}^3(\symb{0}) &} \\
 \CRule{\symb{nest}(x) &}{& \symb{nest}(\symb{nest}(\symb{p}(x))) &}{& x > \symb{s}^3(\symb{0}) &} \\
 \end{array}
 \right\}
 \cup
 \cRsp
\]
	\item another variant of $\cRnest$
\[
 \cRnest'' = 
 \left\{
 \begin{array}{r@{\>}c@{\>}l@{~~}c@{}c@{}c}
 \CRule{\symb{nest}(x,y) &}{& \symb{0} &}{& x \leq \symb{0} &} \\
 \CRule{\symb{nest}(x,y) &}{& \symb{nest}(\symb{nest}(\symb{p}(x),x),y) &}{& x > \symb{0} &} \\
 \end{array}
 \right\}
 \cup
 \cRsp
\]
\end{itemize}
The original PI-based processor $\PIProc$ is efficient but not so powerful.
Though, $\PIProc$ succeeded in proving termination of $\cRack$ while our new PI-based processors failed.
This is because nested recursive call of $\symb{ack}$ does not have to be taken into account to prove termination, and thus, the first argument of $\symb{ack}$, which is not a reducible position, is enough to prove termination. 
On the other hand, we have not introduced the notion of \emph{usable rules} to our implementation, and thus, our new PI-based processors have to take rules in $\cRack$ into account even if we drop all reducible positions of marked symbols by PIs.
Since $\PIProc$ is efficient, we may apply $\PIProc$ and other PI-based processors to a DP problem in order:
if $\PIProc$ does not make the DP problem smaller, then we apply others to the problem.

$\PIProcX[\DecDec]$ and $\PIProcX[\IncDec]$ succeeded in proving termination of $\cRnest$ and $\cRnest''$, but for each of $\cRnest$ and $\cRnest''$, the execution times are quite different, e.g., $\PIProcX[\DecDec]$ took 0.12 and 514 seconds for $\cRnest$ and $\cRnest''$, respectively.
This difference might be caused by the difference of formulas that \textsf{Z3} solved although there are common assignments that satisfies both of the formulas.

\begin{table}[t]
\caption{the result of experiments to prove termination of $\cRMc$, $\cRMc'$, $\cRack$, $\cRack'$, $\cRnest$, $\cRnest'$, and $\cRnest''$}
\label{tbl:comparison}
\small 
\centering
\begin{tabular}{|c|c||c|@{\,}c@{\,}|@{\,}c@{\,}|@{\,}c@{\,}|@{\,}c@{\,}|c|c|}
\hline
\multicolumn{2}{|c||}{} & \multicolumn{5}{|c|}{\textsf{Cter}} & \rule{0pt}{10pt} {\AProVE}~\cite{AProVE} & \textsf{Ctrl}~\cite{KN15lpar} \\
\cline{3-7}
\multicolumn{2}{|c||}{\raisebox{5pt}[0pt]{Example}} & $\PIProc$ & \rule[-2mm]{0pt}{0pt} $\PIProcX[\DecDec]$ & $\PIProcX[\IncDec]$ & $\PIProcX[\DecInc]$ & $\PIProcX[\IncInc]$ & 
{\scriptsize (ver.\ Aug.\ 30, '17)}
	&
{\scriptsize (ver.\ 1.1)}
    \\
\hline\hline
& result & failure & failure & failure & {\Success} & {\Success} & timeout & failure \\
\cline{2-9}
\raisebox{5pt}[0pt]{$\cRMc$} & time (sec.) & 0.08 & 0.60 & 0.56 & 30 & 439 & --- & 0.1 \\
\hline
& result & failure & failure & failure & {\Success} & {\Success} & {\Success} & failure \\
\cline{2-9}
\raisebox{5pt}[0pt]{$\cRMc'$} & time (sec.) & 0.07 & 0.14 & 0.14 & 6.1 & 6.7 & 1.5 & 0.1 \\
\hline
& result & {\Success} & timeout & timeout & failure & failure & {\Success} & {\Success} \\
\cline{2-9}
\raisebox{5pt}[0pt]{$\cRack$} & time (sec.) & 0.18 & --- & --- & 0.28 & 0.31 & 1.6 & 0.2 \\
\hline
& result & {\Success} & failure & failure & failure & failure & {\Success} & {\Success} \\
\cline{2-9}
\raisebox{5pt}[0pt]{$\cRack'$} & time (sec.) & 0.20 & 59 & 134 & 0.31 & 0.31 & 1.4 & 0.2 \\
\hline
& result & failure & {\Success} & {\Success} & failure & failure & {\Success} & failure \\
\cline{2-9}
\raisebox{5pt}[0pt]{$\cRnest$} & time (sec.) & 0.06 & 0.12 & 0.12 & 0.10 & 0.09 & 1.4 & 0.2 \\
\hline
& result & failure & {\Success} & {\Success} & failure & failure & timeout & failure \\
\cline{2-9}
\raisebox{5pt}[0pt]{$\cRnest'$} & time (sec.) & 0.07 & 0.27 & 0.15 & 0.09 & 0.08 & --- & 0.2 \\
\hline
& result & failure & {\Success} & {\Success} & failure & failure & {\Success} & failure \\
\cline{2-9}
\raisebox{5pt}[0pt]{$\cRnest''$} & time (sec.) & 0.08 & 514 & 109 & 0.09 & 0.10 & 1.3 & 0.1 \\
\hline
\end{tabular}
\end{table}

\section{Conclusion}
\label{sec:conclusion}

In this paper, we showed sufficient conditions of PIs for transforming dependency chains into bounded monotone sequences of integers, and improved the PI-based processor proposed in~\cite{SNS11}, providing four PI-based processors.
We showed that two of them are useful to prove termination of a constrained TRS defining the McCarthy 91 function over the integers. 

One of the important related work is the methods in~\cite{FGMSKTZ08} and~\cite{FGPSF09}.
The PI-based processor in~\cite{FGPSF09} for ITRSs is almost the same as $\PIProcX[\DecDec]$, and thus, it cannot prove termination of $\cRMc$.
The PI-based processor in~\cite{FGMSKTZ08} for TRSs uses more general and powerful PIs, and can transform ground rewrite sequences of TRSs into increasing sequences of integers by exchanging the left- and right-hand sides of rewrite rules that appear below \emph{negative contexts} (reducible positions which are given negative coefficients by PIs).
For this reason, our PI-based processors are simplified variants of the PI-based processor while it has to be extended to constrained rewriting. 
The PI-based processor in~\cite{FGMSKTZ08} is not extended to ITRSs in~\cite{FGPSF09}, but the extended processor is implemented in {\AProVE}.
The reason why {\AProVE} failed to prove termination of e.g., $\cRMc$ is that {\AProVE} tries to detect appropriate coefficients for PIs from $-1$ to $2$.
	The range must be enough for many cases, e.g., {\AProVE} succeeded in proving termination of $\cRMc'$.
	By expanding the range of coefficients to e.g., $[-255,256]$, {\AProVE} can immediately prove termination of $\cRMc$.
	This paper showed that the narrow range for coefficients is not enough to prove termination of the McCarthy 91 function.
	
For some examples, the execution time of the proposed processors are larger than we expected, and we would like to improve efficiency. 
Given a DP problem, the current implementation produces a single large quantified non-linear formula of integer arithmetic expressions, and passes the formula to \textsf{Z3} that may spend much time to solve such a complicated formula.
One of our future work is to improve efficiency of the implementation by introducing the way in~\cite[Section~4.1]{FGPSF09} to simplify formulas passed to \textsf{Z3}.

\paragraph{Acknowledgement}
We thank the anonymous reviewers for their useful remarks for further development of our PI-based processors. 
We also thank Carsten Fuhs for his helpful comments to compare our results with~the techniques in~\cite{FGMSKTZ08,FGPSF09}, 
and for his confirming that the expansion of the range for coefficients enables {\AProVE} to succeed in proving termination of the McCarthy 91 function over the integers.


\begin{thebibliography}{10}
\providecommand{\bibitemdeclare}[2]{}
\providecommand{\surnamestart}{}
\providecommand{\surnameend}{}
\providecommand{\urlprefix}{Available at }
\providecommand{\url}[1]{\texttt{#1}}
\providecommand{\href}[2]{\texttt{#2}}
\providecommand{\urlalt}[2]{\href{#1}{#2}}
\providecommand{\doi}[1]{doi:\urlalt{http://dx.doi.org/#1}{#1}}
\providecommand{\bibinfo}[2]{#2}

\bibitemdeclare{inproceedings}{AEFGGLST08}
\bibitem{AEFGGLST08}
\bibinfo{author}{Beatriz \surnamestart Alarc{\'{o}}n\surnameend},
  \bibinfo{author}{Fabian \surnamestart Emmes\surnameend},
  \bibinfo{author}{Carsten \surnamestart Fuhs\surnameend},
  \bibinfo{author}{J{\"{u}}rgen \surnamestart Giesl\surnameend},
  \bibinfo{author}{Ra{\'{u}}l \surnamestart Guti{\'{e}}rrez\surnameend},
  \bibinfo{author}{Salvador \surnamestart Lucas\surnameend},
  \bibinfo{author}{Peter \surnamestart Schneider{-}Kamp\surnameend} \&
  \bibinfo{author}{Ren{\'{e}} \surnamestart Thiemann\surnameend}
  (\bibinfo{year}{2008}): \emph{\bibinfo{title}{Improving Context-Sensitive
  Dependency Pairs}}.
\newblock In \bibinfo{editor}{Iliano \surnamestart Cervesato\surnameend},
  \bibinfo{editor}{Helmut \surnamestart Veith\surnameend} \&
  \bibinfo{editor}{Andrei \surnamestart Voronkov\surnameend}, editors: {\sl
  \bibinfo{booktitle}{Proceedings of the 15th International Conference on Logic
  for Programming, Artificial Intelligence, and Reasoning}}, {\sl
  \bibinfo{series}{Lecture Notes in Computer Science}} \bibinfo{volume}{5330},
  \bibinfo{publisher}{Springer}, pp. \bibinfo{pages}{636--651},
  \doi{10.1007/978-3-540-89439-1\_44}.

\bibitemdeclare{article}{AG00}
\bibitem{AG00}
\bibinfo{author}{Thomas \surnamestart Arts\surnameend} \&
  \bibinfo{author}{J{\"{u}}rgen \surnamestart Giesl\surnameend}
  (\bibinfo{year}{2000}): \emph{\bibinfo{title}{Termination of term rewriting
  using dependency pairs}}.
\newblock {\sl \bibinfo{journal}{Theoretical Computer Science}}
  \bibinfo{volume}{236}(\bibinfo{number}{1-2}), pp. \bibinfo{pages}{133--178},
  \doi{10.1016/S0304-3975(99)00207-8}.

\bibitemdeclare{book}{BN98}
\bibitem{BN98}
\bibinfo{author}{Franz \surnamestart Baader\surnameend} \&
  \bibinfo{author}{Tobias \surnamestart Nipkow\surnameend}
  (\bibinfo{year}{1998}): \emph{\bibinfo{title}{Term Rewriting and All That}}.
\newblock \bibinfo{publisher}{Cambridge University Press},
  \doi{10.1145/505863.505888}.

\bibitemdeclare{inproceedings}{BJ08}
\bibitem{BJ08}
\bibinfo{author}{Adel \surnamestart Bouhoula\surnameend} \&
  \bibinfo{author}{Florent \surnamestart Jacquemard\surnameend}
  (\bibinfo{year}{2008}): \emph{\bibinfo{title}{Automated Induction with
  Constrained Tree Automata}}.
\newblock In \bibinfo{editor}{Alessandro \surnamestart Armando\surnameend},
  \bibinfo{editor}{Peter \surnamestart Baumgartner\surnameend} \&
  \bibinfo{editor}{Gilles \surnamestart Dowek\surnameend}, editors: {\sl
  \bibinfo{booktitle}{Proceedings of the 4th International Joint Conference on
  Automated Reasoning}}, {\sl \bibinfo{series}{Lecture Notes in Computer
  Science}} \bibinfo{volume}{5195}, \bibinfo{publisher}{Springer}, pp.
  \bibinfo{pages}{539--554}, \doi{10.1007/978-3-540-71070-7\_44}.

\bibitemdeclare{inproceedings}{FK08}
\bibitem{FK08}
\bibinfo{author}{Stephan \surnamestart Falke\surnameend} \&
  \bibinfo{author}{Deepak \surnamestart Kapur\surnameend}
  (\bibinfo{year}{2008}): \emph{\bibinfo{title}{Dependency Pairs for Rewriting
  with Built-In Numbers and Semantic Data Structures}}.
\newblock In \bibinfo{editor}{Andrei \surnamestart Voronkov\surnameend},
  editor: {\sl \bibinfo{booktitle}{Proceedings of the 19th International
  Conference on Rewriting Techniques and Applications}}, {\sl
  \bibinfo{series}{Lecture Notes in Computer Science}} \bibinfo{volume}{5117},
  \bibinfo{publisher}{Springer}, pp. \bibinfo{pages}{94--109},
  \doi{10.1007/978-3-540-70590-1\_7}.

\bibitemdeclare{inproceedings}{FK09}
\bibitem{FK09}
\bibinfo{author}{Stephan \surnamestart Falke\surnameend} \&
  \bibinfo{author}{Deepak \surnamestart Kapur\surnameend}
  (\bibinfo{year}{2009}): \emph{\bibinfo{title}{A Term Rewriting Approach to
  the Automated Termination Analysis of Imperative Programs}}.
\newblock In \bibinfo{editor}{Renate~A. \surnamestart Schmidt\surnameend},
  editor: {\sl \bibinfo{booktitle}{Proceedings of the 22nd International
  Conference on Automated Deduction}}, {\sl \bibinfo{series}{Lecture Notes in
  Computer Science}} \bibinfo{volume}{5663}, \bibinfo{publisher}{Springer}, pp.
  \bibinfo{pages}{277--293}, \doi{10.1007/978-3-642-02959-2\_22}.

\bibitemdeclare{inproceedings}{FK12}
\bibitem{FK12}
\bibinfo{author}{Stephan \surnamestart Falke\surnameend} \&
  \bibinfo{author}{Deepak \surnamestart Kapur\surnameend}
  (\bibinfo{year}{2012}): \emph{\bibinfo{title}{Rewriting Induction + Linear
  Arithmetic = Decision Procedure}}.
\newblock In \bibinfo{editor}{Bernhard \surnamestart Gramlich\surnameend},
  \bibinfo{editor}{Dale \surnamestart Miller\surnameend} \&
  \bibinfo{editor}{Uli \surnamestart Sattler\surnameend}, editors: {\sl
  \bibinfo{booktitle}{Proceedings of the 6th International Joint Conference on
  Automated Reasoning}}, {\sl \bibinfo{series}{Lecture Notes in Computer
  Science}} \bibinfo{volume}{7364}, \bibinfo{publisher}{Springer}, pp.
  \bibinfo{pages}{241--255}, \doi{10.1007/978-3-642-31365-3\_20}.

\bibitemdeclare{inproceedings}{FGMSKTZ08}
\bibitem{FGMSKTZ08}
\bibinfo{author}{Carsten \surnamestart Fuhs\surnameend},
  \bibinfo{author}{J{\"{u}}rgen \surnamestart Giesl\surnameend},
  \bibinfo{author}{Aart \surnamestart Middeldorp\surnameend},
  \bibinfo{author}{Peter \surnamestart Schneider{-}Kamp\surnameend},
  \bibinfo{author}{Ren{\'{e}} \surnamestart Thiemann\surnameend} \&
  \bibinfo{author}{Harald \surnamestart Zankl\surnameend}
  (\bibinfo{year}{2008}): \emph{\bibinfo{title}{Maximal Termination}}.
\newblock In \bibinfo{editor}{Andrei \surnamestart Voronkov\surnameend},
  editor: {\sl \bibinfo{booktitle}{Proceedings of the 19th International
  Conference on Rewriting Techniques and Applications}}, {\sl
  \bibinfo{series}{Lecture Notes in Computer Science}} \bibinfo{volume}{5117},
  \bibinfo{publisher}{Springer}, pp. \bibinfo{pages}{110--125},
  \doi{10.1007/978-3-540-70590-1\_8}.

\bibitemdeclare{inproceedings}{FGPSF09}
\bibitem{FGPSF09}
\bibinfo{author}{Carsten \surnamestart Fuhs\surnameend},
  \bibinfo{author}{J{\"{u}}rgen \surnamestart Giesl\surnameend},
  \bibinfo{author}{Martin \surnamestart Pl{\"{u}}cker\surnameend},
  \bibinfo{author}{Peter \surnamestart Schneider{-}Kamp\surnameend} \&
  \bibinfo{author}{Stephan \surnamestart Falke\surnameend}
  (\bibinfo{year}{2009}): \emph{\bibinfo{title}{Proving Termination of Integer
  Term Rewriting}}.
\newblock In \bibinfo{editor}{Ralf \surnamestart Treinen\surnameend}, editor:
  {\sl \bibinfo{booktitle}{Proceedings of the 20th International Conference on
  Rewriting Techniques and Applications}}, {\sl \bibinfo{series}{Lecture Notes
  in Computer Science}} \bibinfo{volume}{5595}, \bibinfo{publisher}{Springer},
  pp. \bibinfo{pages}{32--47}, \doi{10.1007/978-3-642-02348-4\_3}.

\bibitemdeclare{article}{FKN17tocl}
\bibitem{FKN17tocl}
\bibinfo{author}{Carsten \surnamestart Fuhs\surnameend},
  \bibinfo{author}{Cynthia \surnamestart Kop\surnameend} \&
  \bibinfo{author}{Naoki \surnamestart Nishida\surnameend}
  (\bibinfo{year}{2017}): \emph{\bibinfo{title}{Verifying Procedural Programs
  via Constrained Rewriting Induction}}.
\newblock {\sl \bibinfo{journal}{ACM Transactions on Computational Logic}}
  \bibinfo{volume}{18}(\bibinfo{number}{2}), pp. \bibinfo{pages}{14:1--14:50},
  \doi{10.1145/3060143}.

\bibitemdeclare{article}{FNSKS08b}
\bibitem{FNSKS08b}
\bibinfo{author}{Yuki \surnamestart Furuichi\surnameend},
  \bibinfo{author}{Naoki \surnamestart Nishida\surnameend},
  \bibinfo{author}{Masahiko \surnamestart Sakai\surnameend},
  \bibinfo{author}{Keiichirou \surnamestart Kusakari\surnameend} \&
  \bibinfo{author}{Toshiki \surnamestart Sakabe\surnameend}
  (\bibinfo{year}{2008}): \emph{\bibinfo{title}{Approach to Procedural-program
  Verification Based on Implicit Induction of Constrained Term Rewriting
  Systems}}.
\newblock {\sl \bibinfo{journal}{IPSJ Transactions on Programming}}
  \bibinfo{volume}{1}(\bibinfo{number}{2}), pp. \bibinfo{pages}{100--121}.
\newblock \bibinfo{note}{In Japanese (a translated summary is available from
  \url{http://www.trs.css.i.nagoya-u.ac.jp/crisys/})}.

\bibitemdeclare{article}{Giesl97}
\bibitem{Giesl97}
\bibinfo{author}{J{\"{u}}rgen \surnamestart Giesl\surnameend}
  (\bibinfo{year}{1997}): \emph{\bibinfo{title}{Termination of Nested and
  Mutually Recursive Algorithms}}.
\newblock {\sl \bibinfo{journal}{Journal of Automated Reasoning}}
  \bibinfo{volume}{19}(\bibinfo{number}{1}), pp. \bibinfo{pages}{1--29},
  \doi{10.1023/A:1005797629953}.

\bibitemdeclare{inproceedings}{AProVE}
\bibitem{AProVE}
\bibinfo{author}{J{\"u}rgen \surnamestart Giesl\surnameend},
  \bibinfo{author}{Peter \surnamestart Schneider-Kamp\surnameend} \&
  \bibinfo{author}{Ren{\'e} \surnamestart Thiemann\surnameend}
  (\bibinfo{year}{2006}): \emph{\bibinfo{title}{AProVE 1.2: Automatic
  Termination Proofs in the Dependency Pair Framework}}.
\newblock In \bibinfo{editor}{Ulrich \surnamestart Furbach\surnameend} \&
  \bibinfo{editor}{Natarajan \surnamestart Shankar\surnameend}, editors: {\sl
  \bibinfo{booktitle}{Proceedings of the 3rd International Joint Conference on
  Automated Reasoning}}, {\sl \bibinfo{series}{Lecture Notes in Computer
  Science}} \bibinfo{volume}{4130}, \bibinfo{publisher}{Springer}, pp.
  \bibinfo{pages}{281--286}, \doi{10.1007/11814771\_24}.

\bibitemdeclare{inproceedings}{GTS04}
\bibitem{GTS04}
\bibinfo{author}{J{\"{u}}rgen \surnamestart Giesl\surnameend},
  \bibinfo{author}{Ren{\'{e}} \surnamestart Thiemann\surnameend} \&
  \bibinfo{author}{Peter \surnamestart Schneider-Kamp\surnameend}
  (\bibinfo{year}{2005}): \emph{\bibinfo{title}{The Dependency Pair Framework:
  Combining Techniques for Automated Termination Proofs}}.
\newblock In \bibinfo{editor}{Franz \surnamestart Baader\surnameend} \&
  \bibinfo{editor}{Andrei \surnamestart Voronkov\surnameend}, editors: {\sl
  \bibinfo{booktitle}{Proceedings of the 11th International Conference on Logic
  for Programming, Artificial Intelligence, and Reasoning}}, {\sl
  \bibinfo{series}{Lecture Notes in Computer Science}} \bibinfo{volume}{3452},
  \bibinfo{publisher}{Springer}, pp. \bibinfo{pages}{301--331},
  \doi{10.1007/978-3-540-32275-7\_21}.

\bibitemdeclare{article}{GTSF06}
\bibitem{GTSF06}
\bibinfo{author}{J{\"{u}}rgen \surnamestart Giesl\surnameend},
  \bibinfo{author}{Ren{\'{e}} \surnamestart Thiemann\surnameend},
  \bibinfo{author}{Peter \surnamestart Schneider{-}Kamp\surnameend} \&
  \bibinfo{author}{Stephan \surnamestart Falke\surnameend}
  (\bibinfo{year}{2006}): \emph{\bibinfo{title}{Mechanizing and Improving
  Dependency Pairs}}.
\newblock {\sl \bibinfo{journal}{Journal of Automated Reasoning}}
  \bibinfo{volume}{37}(\bibinfo{number}{3}), pp. \bibinfo{pages}{155--203},
  \doi{10.1007/s10817-006-9057-7}.

\bibitemdeclare{inproceedings}{Kop13}
\bibitem{Kop13}
\bibinfo{author}{Cynthia \surnamestart Kop\surnameend} (\bibinfo{year}{2013}):
  \emph{\bibinfo{title}{Termination of {LCTRSs} (extended abstract)}}.
\newblock In: {\sl \bibinfo{booktitle}{Proceedings of the 13th International
  Workshop on Termination}}, pp. \bibinfo{pages}{1--5}.
\newblock
  \urlprefix\url{http://www.imn.htwk-leipzig.de/WST2013/papers/paper\_12.pdf}.

\bibitemdeclare{inproceedings}{KN13frocos}
\bibitem{KN13frocos}
\bibinfo{author}{Cynthia \surnamestart Kop\surnameend} \&
  \bibinfo{author}{Naoki \surnamestart Nishida\surnameend}
  (\bibinfo{year}{2013}): \emph{\bibinfo{title}{Term Rewriting with Logical
  Constraints}}.
\newblock In \bibinfo{editor}{Pascal \surnamestart Fontaine\surnameend},
  \bibinfo{editor}{Christophe \surnamestart Ringeissen\surnameend} \&
  \bibinfo{editor}{Renate~A. \surnamestart Schmidt\surnameend}, editors: {\sl
  \bibinfo{booktitle}{Proceedings of the 9th International Symposium on
  Frontiers of Combining Systems}}, {\sl \bibinfo{series}{Lecture Notes in
  Artificial Intelligence}} \bibinfo{volume}{8152}, pp.
  \bibinfo{pages}{343--358}, \doi{10.1007/978-3-642-40885-4\_24}.

\bibitemdeclare{inproceedings}{KN14aplas}
\bibitem{KN14aplas}
\bibinfo{author}{Cynthia \surnamestart Kop\surnameend} \&
  \bibinfo{author}{Naoki \surnamestart Nishida\surnameend}
  (\bibinfo{year}{2014}): \emph{\bibinfo{title}{Automatic Constrained Rewriting
  Induction towards Verifying Procedural Programs}}.
\newblock In \bibinfo{editor}{Jacques \surnamestart Garrigue\surnameend},
  editor: {\sl \bibinfo{booktitle}{Proceedings of the 12th Asian Symposium on
  Programming Languages and Systems}}, {\sl \bibinfo{series}{Lecture Notes in
  Computer Science}} \bibinfo{volume}{8858}, pp. \bibinfo{pages}{334--353},
  \doi{10.1007/978-3-319-12736-1\_18}.

\bibitemdeclare{inproceedings}{KN15lpar}
\bibitem{KN15lpar}
\bibinfo{author}{Cynthia \surnamestart Kop\surnameend} \&
  \bibinfo{author}{Naoki \surnamestart Nishida\surnameend}
  (\bibinfo{year}{2015}): \emph{\bibinfo{title}{{Constrained Term Rewriting
  tooL}}}.
\newblock In \bibinfo{editor}{Martin \surnamestart Davis\surnameend},
  \bibinfo{editor}{Ansgar \surnamestart Fehnker\surnameend},
  \bibinfo{editor}{Annabelle \surnamestart McIver\surnameend} \&
  \bibinfo{editor}{Andrei \surnamestart Voronkov\surnameend}, editors: {\sl
  \bibinfo{booktitle}{Proceedings of the 20th International Conference on Logic
  for Programming, Artificial Intelligence, and Reasoning}}, {\sl
  \bibinfo{series}{Lecture Notes in Computer Science}} \bibinfo{volume}{9450},
  pp. \bibinfo{pages}{549--557}, \doi{10.1007/978-3-662-48899-7\_38}.

\bibitemdeclare{inproceedings}{LM14}
\bibitem{LM14}
\bibinfo{author}{Salvador \surnamestart Lucas\surnameend} \&
  \bibinfo{author}{Jos{\'{e}} \surnamestart Meseguer\surnameend}
  (\bibinfo{year}{2014}): \emph{\bibinfo{title}{{2D} Dependency Pairs for
  Proving Operational Termination of {CTRSs}}}.
\newblock In \bibinfo{editor}{Santiago \surnamestart Escobar\surnameend},
  editor: {\sl \bibinfo{booktitle}{Proceedings of the 10th International
  Workshop on Rewriting Logic and Its Applications}}, {\sl
  \bibinfo{series}{Lecture Notes in Computer Science}} \bibinfo{volume}{8663},
  \bibinfo{publisher}{Springer}, pp. \bibinfo{pages}{195--212},
  \doi{10.1007/978-3-319-12904-4\_11}.

\bibitemdeclare{inproceedings}{z3}
\bibitem{z3}
\bibinfo{author}{Leonardo~Mendon{\c{c}}a \surnamestart de~Moura\surnameend} \&
  \bibinfo{author}{Nikolaj \surnamestart Bj{\o}rner\surnameend}
  (\bibinfo{year}{2008}): \emph{\bibinfo{title}{{Z3}: An Efficient {SMT}
  Solver}}.
\newblock In \bibinfo{editor}{C.~R. \surnamestart Ramakrishnan\surnameend} \&
  \bibinfo{editor}{Jakob \surnamestart Rehof\surnameend}, editors: {\sl
  \bibinfo{booktitle}{Proceedings of the 14th International Conference on Tools
  and Algorithms for the Construction and Analysis of Systems}}, {\sl
  \bibinfo{series}{Lecture Notes in Computer Science}} \bibinfo{volume}{4963},
  \bibinfo{publisher}{Springer}, pp. \bibinfo{pages}{337--340},
  \doi{10.1007/978-3-540-78800-3\_24}.

\bibitemdeclare{book}{Ohl02}
\bibitem{Ohl02}
\bibinfo{author}{Enno \surnamestart Ohlebusch\surnameend}
  (\bibinfo{year}{2002}): \emph{\bibinfo{title}{Advanced Topics in Term
  Rewriting}}.
\newblock \bibinfo{publisher}{Springer}, \doi{10.1007/978-1-4757-3661-8}.

\bibitemdeclare{inproceedings}{SNS11}
\bibitem{SNS11}
\bibinfo{author}{Tsubasa \surnamestart Sakata\surnameend},
  \bibinfo{author}{Naoki \surnamestart Nishida\surnameend} \&
  \bibinfo{author}{Toshiki \surnamestart Sakabe\surnameend}
  (\bibinfo{year}{2011}): \emph{\bibinfo{title}{On Proving Termination of
  Constrained Term Rewrite Systems by Eliminating Edges from Dependency
  Graphs}}.
\newblock In \bibinfo{editor}{Herbert \surnamestart Kuchen\surnameend}, editor:
  {\sl \bibinfo{booktitle}{Proceedings of the 20th International Workshop on
  Functional and (Constraint) Logic Programming}}, {\sl
  \bibinfo{series}{Lecture Notes in Computer Science}} \bibinfo{volume}{6816},
  \bibinfo{publisher}{Springer}, pp. \bibinfo{pages}{138--155},
  \doi{10.1007/978-3-642-22531-4\_9}.

\bibitemdeclare{article}{SNSSK09}
\bibitem{SNSSK09}
\bibinfo{author}{Tsubasa \surnamestart Sakata\surnameend},
  \bibinfo{author}{Naoki \surnamestart Nishida\surnameend},
  \bibinfo{author}{Toshiki \surnamestart Sakabe\surnameend},
  \bibinfo{author}{Masahiko \surnamestart Sakai\surnameend} \&
  \bibinfo{author}{Keiichirou \surnamestart Kusakari\surnameend}
  (\bibinfo{year}{2009}): \emph{\bibinfo{title}{Rewriting Induction for
  Constrained Term Rewriting Systems}}.
\newblock {\sl \bibinfo{journal}{IPSJ Transactions on Programming}}
  \bibinfo{volume}{2}(\bibinfo{number}{2}), pp. \bibinfo{pages}{80--96}.
\newblock \bibinfo{note}{In Japanese (a translated summary is available from
  \url{http://www.trs.css.i.nagoya-u.ac.jp/crisys/})}.

\bibitemdeclare{inproceedings}{Toy87}
\bibitem{Toy87}
\bibinfo{author}{Yoshihito \surnamestart Toyama\surnameend}
  (\bibinfo{year}{1987}): \emph{\bibinfo{title}{Confluent Term Rewriting
  Systems with Membership Conditions}}.
\newblock In \bibinfo{editor}{St{\'e}phane \surnamestart Kaplan\surnameend} \&
  \bibinfo{editor}{Jean-Pierre \surnamestart Jouannaud\surnameend}, editors:
  {\sl \bibinfo{booktitle}{Proceedings of the 1st International Workshop on
  Conditional Term Rewriting Systems}}, {\sl \bibinfo{series}{Lecture Notes in
  Computer Science}} \bibinfo{volume}{308}, \bibinfo{publisher}{Springer}, pp.
  \bibinfo{pages}{228--241}, \doi{10.1007/3-540-19242-5\_17}.

\bibitemdeclare{incollection}{Zan03}
\bibitem{Zan03}
\bibinfo{author}{Hans \surnamestart Zantema\surnameend} (\bibinfo{year}{2003}):
  \emph{\bibinfo{title}{Termination}}.
\newblock In: {\sl \bibinfo{booktitle}{{Term Rewriting Systems}}},
  chapter~\bibinfo{chapter}{6}, {\sl \bibinfo{series}{Cambridge Tracts in
  Theoretical Computer Science}}~\bibinfo{volume}{55},
  \bibinfo{publisher}{Cambridge University Press}, pp.
  \bibinfo{pages}{181--259}, \doi{10.1017/S1471068405222445}.

\end{thebibliography}

\end{document}